\documentclass[10pt, sigconf, nonacm]{acmart}
\usepackage[utf8]{inputenc}
\usepackage[english]{babel}

\renewcommand\footnotetextcopyrightpermission[1]{}

\usepackage{url}

\usepackage{graphicx}
\usepackage[utf8]{inputenc}

\usepackage{comment}
\usepackage{algorithm}

\usepackage{algpseudocode}

\usepackage{amsfonts} 
\usepackage{array}
\usepackage{float}
\usepackage{caption}
\usepackage{subcaption}

\usepackage{enumitem}

\usepackage{xcolor}

\usepackage{xspace}
\newcommand{\Pipelet}[0]{\textsc{Pipelet}\xspace}

\newcommand{\Timeout}[0]{\textsc{Timeout}\xspace}
\newcommand{\Votes}[0]{\textsc{Votes}\xspace}
\newcommand{\Sync}[0]{\textsc{Sync}\xspace}

\newcommand{\Vote}[0]{\textsc{Vote}\xspace}
\newcommand{\Proposal}[0]{\textsc{Proposal}\xspace}
\newcommand{\Signatures}[0]{\textsc{Signatures}\xspace}
\newcommand{\Signature}[0]{\textsc{Signature}\xspace}

\newcommand{\chain}[0]{\ensuremath{\mathcal{C}}}

\usepackage[capitalise]{cleveref}

\usepackage[textsize=tiny, textwidth=1.1cm, disable]{todonotes}

\newcommand{\vivek}[1]{\todo[color=green!10, linecolor=green!50!black]{\textbf{vivek:} #1}}

\newtheorem{theorem}{Theorem}

\newtheorem{definition}{Definition}

\begin{document}
\pagestyle{plain}

\title{Pipelet: Practical Streamlined Blockchain Protocol}
\author{Vivek Karihaloo}
\author{Ruchi Shah}
\author{Panruo Wu}
\affiliation{University of Houston}
\author{Aron Laszka}
\affiliation{Pennsylvania State University}

\begin{abstract}
Fueled by the growing popularity of proof-of-stake blockchains, there has been increasing interest and progress in permissioned consensus protocols, which could provide a simpler alternative to existing protocols, such as Paxos and PBFT.
In particular, the recently proposed Streamlet protocol provides a surprisingly simple and streamlined consensus approach, which crystallizes years of research in simplifying and improving classical consensus protocols.
While the simplicity of Streamlet is a major accomplishment, the protocol lacks certain practical features, such as supporting a stable block proposer, and it makes strong assumptions, such as synchronized clocks and the implicit echoing of all messages.
Most importantly, it requires sending $O(N^3)$ messages per block in a network of $N$ nodes, which poses a significant challenge to its application in larger networks.
To address these limitations, we introduce Pipelet, a practical streamlined consensus protocol.
Pipelet employs the same block-finalization rule as Streamlet, but attains state-of-the-art performance in terms of communication complexity and provides features that are crucial for practical applications, such as clock synchronization and stable block proposers.
At the same time, Pipelet retains the simplicity of Streamlet, which presents significant practical advantages, such as ease of implementation and verification.   
\end{abstract}

\maketitle

\section{Introduction}
\label{sec:intro}

In an increasingly networked digital world,  society is becoming more and more dependent on reliable digital services in spite of faulty components.
Furthermore, as networked systems become more decentralized and distributed, they must also tolerate malicious---or ``Byzantine''---components while still providing reliable service.
Over the last decade, the explosion of blockchain technology has opened possibilities of open, permissionless, and distributed digital infrastructure for a variety of applications, including finance~\cite{werner2021sok}, supply-chain management~\cite{chang2020blockchain}, and more~\cite{laszka2017providing,eisele2020blockchains}.
As blockchain applications and their adoption grow, the demand for more energy efficient and high performance blockchain technology increases tremendously.
At the core of fault tolerant distributed systems and blockchain technology is the fundamental problem of  reaching consensus, that is, reaching agreement among honest participants despite the presence of malicious ones.

Distributed consensus can be viewed as a replicated state machine~\cite{schneider1990implementing}, where a deterministic state machine is replicated across a set of processes but functions as one.
The input of the state machine is called transaction, which causes a transition of states.
A transaction~\cite{gray1981transaction} is an atomic operation that either completes or does not occur at all.
The responsibility of a consensus protocol is to order the transactions such that the resulting transaction log of every process is the~same.

Despite the great success of the Bitcoin network, there are major concerns, including high energy consumption and environmental impact, slow transaction confirmation, and low throughput.
More modern blockchain technology often builds on Proof-of-Stake~\cite{daian2019snow,kiayias2017ouroboros} instead of Proof-of-Work to be energy efficient, which often relies on a permissioned committee of validators~\cite{ethereum_pos}.
Such approaches can be described as ``blockchainized'' classic consensus algorithms, such as Tendermint~\cite{buchman2016tendermint,buchman2018latest}, which adopts classic consensus algorithm PBFT~\cite{castro1999practical}.
This approach typically involves two rounds of voting to confirm a block: a prepare round and a commit round.
An interesting recent idea is to pipeline the voting, that is, to piggyback the second round (commit) of the block on the first round (prepare) of the next block.
HotStuff~\cite{yin2019hotstuff}, PaLa~\cite{chan2018pala}, and Casper FFG~\cite{buterin2017casper} are such pipelined protocols.
On the other hand, an extremely simple and elegant protocol called Streamlet~\cite{chan2020streamlet} emerged under the assumption of synchronous~network.

Inspired by these two lines of work, we introduce a high-performance, simple, partially synchronized consensus protocol called \Pipelet.
\Pipelet is a propose-vote, longest-chain blockchain protocol, which employs the same block finalization rule as Streamlet~\cite{chan2020streamlet}.
However, unlike Streamlet, \Pipelet does not assume that all messages are implicitly echoed.
As a result, the communication complexity of \Pipelet in the non-failure scenario is only $O(N)$ messages per finalized block, while the complexity of Streamlet is $O(N^3)$.
Further, in contrast to Streamlet, \Pipelet provides clock synchronization and stable proposers---inspired by PaLa~\cite{chan2018pala}---to enable practical applications of the protocol.
To support these features and relax the assumption of implicit message echoing, \Pipelet includes a simple timeout-sync mechanism, retaining most of the attractive simplicity of Streamlet.
This simplicity stands in stark contrast with many state-of-the-art protocols, which provide similar communication complexity at the cost of complex protocol designs.
We believe that simple protocol designs have significant advantages, ranging from ease of theoretical analysis to ease of implementation, with smaller codebases and fewer potential~bugs.

\subsection{Overview of \Pipelet}

Here, we provide a very brief overview of the key elements of \Pipelet; we will provide a formal specification in \cref{sec:protocol}, once we have introduced the necessary assumptions and terminology.
The objective of the \Pipelet protocol is to enable participating nodes to reach consensus on a chain of finalized transaction blocks, which are recorded irrevocably and immutably.

\noindent \Pipelet operates in time periods, called epochs, as follows:

\begin{itemize}%
    \item \emph{Block Proposal:} In each epoch, one proposer is eligible to propose new blocks, chosen by applying a round robin policy to the set of nodes. The eligible proposer can create new blocks that extend a longest notarized chain of blocks and broadcast them to all other nodes.
    Such a proposal also includes a notarization of the existing block that the proposed block extends. 
    \item \emph{Voting:} Upon receiving a new block from the eligible proposer of the epoch, a  node verifies the validity of the block (e.g., it has to extend a longest notarized blockchain, taking into account the notarization that comes ``piggyback'' with the proposal).
    If the node finds that the block is valid, it responds by sending its vote for the block back to the proposer.
    \item \emph{Notarization:}
    Once a node has received votes for a block from two thirds of the nodes, the node considers the block to be notarized.
    Each node maintains its own set of notarized blocks, which is based on the votes that the node has seen, and its own set of longest blockchains formed by these notarized blocks.
    Note that nodes do not need to reach consensus on the set of notarized blocks; they need to reach consensus only on the finalized ones.
    \item \emph{Finalization:} If three consecutive blocks of a chain---without any timeouts between them---are notarized, then the node considers the middle block as well as all the preceding blocks of the chain to be finalized.
    These finalized blocks are irrevocably and immutably recorded.
    \item \emph{Epoch Advancement:}
    If a  node has not seen a new notarized block on the longest chain in a certain amount of time, it broadcasts a timeout message to all other nodes. The timeout message contains signatures on the epoch to which the node wishes to advance.
    Once a node has received signatures from two thirds  of the nodes, indicating their intention to move to next epoch, the node advances to the next epoch, where a different node becomes the eligible proposer.
    These timeouts are crucial for preventing a faulty or malicious proposer from obstructing the protocol.
    \item \emph{Block Synchronization:}
    Finally, when a nodes receives a timeout message,  it also broadcasts a block synchronization message. This synchronization message contains a set of nodes from the longest notarized chains, which prevents a malicious or faulty proposer from ``splitting'' the network by sending piggyback notarizations to a subset of the nodes. 
\end{itemize}

The remainder of this paper is organized as follows.
\cref{sec:related} provides an overview on related work in the area of consensus protocols.
\cref{sec:problem} introduces the assumptions and objectives for our protocol.
\cref{sec:dataStructDescp} describes the data structures and messages of the \Pipelet protocol.
\cref{sec:protocol} presents a formal specification of the \Pipelet protocol.
\cref{sec:analysis} states our main theorems on the consistency and liveness guarantees of our protocol.
\cref{perf} provides experimental results comparing \Pipelet with Pala and Streamlet.
\cref{sec:concl} provides concluding~remarks.
\cref{CP,LP} prove our consistency and liveness theorems, building on a series of lemmas.
\cref{sec:formula} provides detailed description of how we computed the messages in \cref{tab:comparison}.

 \vivek{The "C" in PBFT is a a constant where to mention about it?}

\begin{table*}[h!]
\centering
\caption{Comparison of Number of Messages}
\label{tab:comparison}
\begin{tabular}{ | c | l | l | } 
  \hline
 Consensus Protocol & Normal Scenario & Failure Scenario  \\ 
  \hline\hline
  Chained HotStuff~\cite{yin2018hotstuff} & $2N-f$ & $f\cdot(3N-f-1)+2N-f $  \\ 
  \hline
  Sync HotStuff~\cite{abraham2020sync} & $2(N-1)^2 $ & $f \cdot((N-1)(5N-3) )  + 2(N-1)^2$  \\ 
  \hline

  PBFT~\cite{castro1999practical} & $N^2+2N(f+1) - 1$ & $f\cdot(2f + 2f(N)) + $C$ \cdot N(2f+1) + N^2+2N(f+1) - 1 $ \\ 
  \hline 

  Tendermint~\cite{buchman2016tendermint} & $N + 4f$ & $f\cdot(N+4f)+N+4f$ \\ 
  \hline
  EBFT~\cite{ic_cyber_2019} &  $2(N+f) - 1$ & $f\cdot(2f(N+2))+2(N+f) - 1$ + $k \cdot (N-1)$\\ 
  \hline
  \Pipelet &$ 2N-2$ & $f \cdot (N^2+N-2) + 2N-2$\\ 
  \hline
\end{tabular}
\end{table*}

\section{Related Work}
\label{sec:related}

A blockchain, or consensus protocol, seeks to agree on an ever growing linearly ordered
log, such that consistency and liveness are ensured. Historically, there are three
main flavors of consensus protocols: Nakamoto consensus, classic consensus, 
and DAG consensus. \emph{Nakamoto consensus}, named after Bitcoin's inventor
Satoshi Nakamoto~\cite{nakamoto2008bitcoin}, chooses the longest chain fork rule to reach consensus at
high probability. It is elegant and robust, and has been empirically validated
by the Bitcoin network which secured almost one trillion U.S. dollar value for a decade. 
A particular advantage of Nakamoto consensus is that it is completely permissionless;
any node can join or leave at anytime, and the network is particularly robust
to node churn. Unfortunately, Nakamoto consensus has also been known to be
slow (requiring multiple confirmations to reach high probability), and the low
throughput is inherent~\cite{pass2017analysis} for its security.
\emph{Classic consensus} protocols
have deep roots in distributed systems research. Consensus is reached immediately
and deterministically by voting. It is not possible to have forking in classic
consensus protocols applied in blockchain. Notable examples of classic Byzantine fault tolerant consensus
include Paxos~\cite{lamport2001paxos}, PBFT~\cite{castro1999practical}, and
Tendermint~\cite{buchman2016tendermint}.
Paxos was one of the first Byzantine fault tolerant protocols, but
it was complicated. Later projects, such as Raft~\cite{ongaro2014search},
attempt to simplify the protocol and make it easy to understand. 
Classic consensus protocols reach
deterministic consensus much faster than Nakamoto consensus; however, they are
not permissionless because nodes cannot join or leave at any time. Also, the
voting committee has to be much smaller than that of the Nakamoto consensus, making it
more centralized. \emph{DAG consensus} is less relevant to this paper, so
we will omit it from our discussion of related work. 

Blockchains need some node to propose a block and
order the transactions. In a \emph{Proof-of-Work} (PoW) blockchain such as
Bitcoin, the miner that first solves a difficult puzzle gets to propose
a block and obtain a reward. As it is very expensive to solve the puzzle, the longest
chain has the most work performed; thus, it is the most expensive to fork. This
creates economic incentive for the network to be consistent. However, PoW
requires prohibitive amounts of energy. An alternative is \emph{Proof-of-Stake} (PoS)
consensus. In a very simplistic PoS, the mining power is replaced by virtual mining power
proportional to the currency that a miner has. The canonical chain is
still chosen to be the longest chain. This simplistic PoS protocol
is vulnerable to nothing-at-stake attacks---where miners mine on multiple forks
at the same time, which is always more profitable and incurs no extra cost---and other vulnerabilities. A more robust PoS protocol is based on a permissioned
committee of voters, and often employs classic consensus instead of
Nakamoto consensus.  \Pipelet, the  consensus protocol proposed in this paper,
is interesting in that although it is permissioned and can
reach deterministic consensus, it always extends the longest chain. 
This is in contrast to classic consensus protocols that do not have forks, 
and also in contrast to PaLa~\cite{chan2018pala}, which picks the ``freshest'' chain to extend.

\paragraph{Permissioned Blockchains}

Chan et al.\ proposed a simple partially-synchronous blockchain protocol called PaLa, which was inspired
by the pipelined-BFT par\-a\-digm \cite{chan2018pala}. 
They also proposed a conceptually simple and provably secure committee rotation algorithm for
PaLa as well as a generalization called ``doubly-pipelined PaLa'' that is geared towards
settings that require high throughput.
In a follow-up paper,
Chan and Shi
introduced an extremely simple and natural paradigm, called Streamlet, for
constructing consensus protocols~\cite{chan2020streamlet}. These protocols are inspired by the core techniques that have
been uncovered in the preceding years; however, their
proposal is simpler than any protocol before. 

One drawback of PaLa is that it uses the notion of ``freshest chain'' for consensus, which is more complex than the usual notion of ``longest chain.''
Interestingly, it can be shown that every freshest chain in the PaLa protocol (under its assumptions) is actually a longest chain, which we formalize in the following lemma.

\paragraph{Consensus Protocol Performance}
Although there has been significant research done in this area of permissioned consensus protocols, analyzing the performance of these protocols has received relatively very little attention. In 2017, Dinh et al. proposed a framework, called as BLOCKBENCH~\cite{acm_sigmod2017}, to evaluate blockchain systems that were designed with stronger security assumptions. The framework evaluates performance in terms of throughput, latency, scalability and fault-tolerance under varying number of nodes. It successfully identified gaps in performance among the three systems, which were attributed to the choice of design at various layers of the software stack.  This study was later complemented by Pongnumkul et al., who performed detailed  analysis of two well known blockchain platforms, Hyperledger Fabric and Ethereum~\cite{icccn_ieee2017}. The goal of this analysis was twofold; one was to benchmark these state-of-art systems, and the second was to enable companies to make informed decision about adopting blockchain technology in their IT systems.

 On the other hand, Hao et all.\ evaluated consensus algorithms with up to 100,000 transactions in Ethereum and Hyperledger Fabric and analyzed the results in terms of latency and throughput~\cite{ivs_ieee_2018}. {Ethereum implemented a variant of Proof-of-Work (PoW) consensus algorithm} using the same underlying principles as PoW, whereas  {Hyperledger Fabric adapted  Practical Byzantine Fault Tolerance (PBFT)}. When running workloads of the same configuration, the average latency of Ethereum showed {16.55x} higher throughput than that of Hyperledger Fabric, {4.13x}. When running workloads of different configurations, the latency of PoW rapidly increased as compared to PBFT. 
 
 In 2019, Li at al. proposed a new Extensible Consensus Algorithm based on PBFT (E-PBFT)~\cite{ic_cyber_2019} that was focused on solving the limitations of PBFT like large communication overhead, unable to identify Byzantine nodes and not suitable for dynamic networks. Due to the core design attributes of a classic consensus protocol, the need to trust any node does not exist. Alternately, the data recorded in the network can be trusted. This protocol implemented the selection of consensus nodes with verifiable random function (VRF) for dynamic networks. When the network is stable, lesser communication was required to reach consensus.

A more recent paper~\cite{csDC_2021} studied the performance and scalability of prominent consensus protocols like Practical Byzantine Fault Tolerance (PBFT)~\cite{castro1999practical}, Tendermint~\cite{buchman2016tendermint}, HotStuff~\cite{yin2019hotstuff} and Streamlet~\cite{chan2020streamlet}. When the protocols were evaluated on Amazon Web Services (AWS) under identical conditions, they revealed certain limitations due to the complexity in communication. Both, Streamlet and Pala use implicit echoing which makes them difficult to use for larger networks. Although Paxos was complicated, it could reach consensus in an asynchronous network that is prone to crash failures as oppose to the byzantine model that requires a number of cryptographic operations. PBFT was one of the first practical BFT replication solutions to the Byzantine problem. Empirical results have proved that HotStuff performed significantly better than all other protocols due to its design attributes like replacing all-to-all with all-to-one communication.  The proposed protocol, \Pipelet is also inspired by these core design attributes and ensures
ordered communication across all honest nodes to improve throughput.

One more interesting consensus algorithm, proposed by Abraham in 2020 is Sync Hotstuff~\cite{abraham2020sync}. It is a synchronous BFT protocol which provides reasonable throughput comparable to partially synchronous consensus protocols. 

In \cref{tab:comparison}, \vivek{does this look okay to you} 
 we present the number of messages required to finalize a block under normal conditions and failure conditions. By normal conditions, we mean when all nodes follow the protocol; and by failure condition, we mean when dishonest nodes do not follow the protocol and the block finalization does not happen. In \cref{tab:comparison} for the failure scenario we have considered consecutive $|f|$ faulty proposers, who just act dishonest just before the block is about to get finalized. Detailed explanation of how we computed the messages in \cref{tab:comparison} is provided in \cref{sec:formula}. Further, in \cref{tab:comparison}, $f$ represents the maximum number of dishonest nodes allowed by the respective consensus algorithm, $N$ represents the total number of nodes and $C \in \mathbb N_{> 0}$ and $k \in \mathbb N_{> 0}$ represents a positive integer as mentioned in \cref{{sec:formula}}. Under normal scenario, \Pipelet only requires $2N-2$ messages for block finalization which is less than the number of messages required for a block finalization by consensus algorithms mentioned in \cref{tab:comparison}. Also, please note under failure scenario number of messages required to finalize a block is comparable to other consensus algorithms mentioned in \cref{tab:comparison}. Under failure scenario, among the consensus algorithms mentioned in \cref{tab:comparison}, only HotStuff~\cite{yin2018hotstuff} and Tendermint~\cite{buchman2016tendermint} require less messages for block finalization as compared to \Pipelet.   \vivek{does this look okay to you}

\section{Assumptions, Objectives, and Threat Model}
\label{sec:problem}

\subsection{Assumptions}
\label{sec:setting_assumption}

In this section, we introduce our assumptions and objectives for the \Pipelet protocol. 
Since \Pipelet is based on PaLa, our assumptions are similar to those of PaLa, the key difference being (a) the assumption of ordered point-to-point message delivery, which obviates the need for nodes to store messages and process them later when they become valid, (b) use of the longest notarized chain instead of freshest notarized chain by the honest proposer to propose blocks, (c) and nodes do not implicitly echo messages, (d) the $\Votes$ do not get broadcasted.
In \cref{pc}, we describe the two roles that nodes play in the \Pipelet protocol and our assumptions about the nodes' behavior. In \cref{cp}, we introduce our assumptions about cryptographic primitives. In \cref{ca}, we discuss assumptions related to timing and communications. 
Finally, in \cref{obj}, we introduce the objectives of \emph{consistency} and \emph{liveness}. For a list of notations used in this paper, please see \cref{tab:notation} in \cref{sec:notation}.

\subsubsection{Proposer and Voting Nodes} \label{pc}
We assume that each node can play two type of roles: \emph{proposer role} and \emph{voting role}. Proposer node is responsible for creating blocks, while voter nodes are responsible for verifying the validity of these blocks and voting on them. 
Let $N$ denote the set of voting nodes and one of them in every epoch acts as a proposer. The proposer node plays dual role of being a proposer and also a voting node.
We say that a node is \emph{honest} if it always follows the protocol.
We assume that there are more than $\frac{2}{3} |N|$ honest nodes (i.e., more than two thirds of the voting nodes $N$ are honest).
We make no assumptions about the behavior of other nodes.

\subsubsection{Cryptographic Primitives} 
\label{cp} 

We assume that there exists a \emph{digital-signature scheme}, and each node has its own public-private key pair and knows the public key of every other node in $N$ (e.g., based on a public-key infrastructure).
In our pseudo-code, we let functions $\textbf{sign}(\cdot)$ and $\textbf{verify}(\cdot)$ denote the creation and verification of digital signatures.
Further, we assume that the digital-signature scheme supports signature aggregation: given a set of signatures from different nodes, it is possible to aggregate all these signatures into a single short signature.
In our pseudo-code, we let operator \textbf{aggregate} denote the aggregation of signatures.
Note that when a node verifies an aggregated signature in \Pipelet, it only needs to verify if at least $\frac{2}{3} |N|$ nodes have signed (denoted as $\textbf{verify}\left(|\text{set of aggregated signatures} \geq \frac{2}{3} |N| \right)$).
In practice, we can implement this using, e.g., Boneh – Lynn – Shacham (BLS) signature  aggregation~\cite{boneh2003aggregate}. 
\vivek{latex issue, not sure how to fix it}

Second, we assume that there exists a \emph{hash function} $H$.
We assume that both the digital-signature scheme and the hash function are secure in the sense that signature forgeries and hash collisions are impossible.
Please note that it would be straightforward to generalize our results by considering the possibility of forgeries and collisions.
We would need to extend each statement with an additional clause that considers the probability of forgeries and collisions.
Since we focus on the consensus protocol in this paper, we disregard the orthogonal matter of forgeries and collisions for ease of exposition.

\subsubsection{Timing and Communications} 
\label{ca}
We assume that time is continuous, and each node has a clock that can measure elapsed real time. 
We let $\Delta$ denote the unit of time, which is defined as the time bound on message delivery during a period synchrony; see detailed definition below. For the rest of the paper, we assume that $1sec \geq 5\Delta$ and $1min \geq 30\Delta$. Please note $sec$ and $min$ do not represent wall-clock seconds and minutes;
they are just convenient aliases for units of time measurement. %

Next, we introduce two assumptions about how nodes communicate, named \emph{ordered point-to-point message delivery}, which are crucial for ensuring consistency.
Ordered point-to-point message delivery means that from one node to another, messages are received in the same order as they are sent.
More formally, if node $n_1$ first sends message $m_1$ to node $n_2$ and later sends message $m_2$, then node $n_2$ will see message $m_1$ before message $m_2$.
This assumption is easy to satisfy in practice (e.g., using sequence numbers), and the underlying communication layer may provide it naturally.

Finally, we introduce 
assumptions about the delivery of messages, which are crucial for ensuring liveness.
Specifically, we assume a \emph{partially synchronous model for communication}. 
In a synchronous model, messages have a known bound on delivery
delay, which can be very difficult to guarantee in practical communication networks. On the other hand, an asynchronous network precludes
reaching consensus in a deterministic way
in the presence
of a single faulty node~\cite{fischer1985impossibility}.
A more practical
assumption is the partially synchronous model
\cite{dwork1988consensus}. In \cite{dwork1988consensus},
two partially synchronous models are proposed: a
model with unknown  delay bound, and a model
with period of synchrony. From a feasibility perspective,
they are equivalent. In this paper, we use the 
period of synchrony model as it is commonly used
in blockchains, where the participants need to
reach consensus continuously. 

We define the notion of \emph{period of synchrony} as follows.

\begin{definition}[Period of Synchrony] 
A time period $[t_0, t_1]$ is a \emph{period of synchrony} if and only if a message sent by an honest node at its local time $t \leq t_1$ as observed by the nodes' clock will be received by every honest recipient at its local time(observed by their clocks) by $\max\{t_0, t + \Delta\}$ at latest.
\end{definition}

\subsection{Objectives} \label{obj}

Our objective is to provide \emph{consistency} and \emph{liveness}. To define these terms, we consider the purpose of the blockchain protocol, which is to record transactions and to ensure that recorded transactions are immutable and irrevocable.
We assume that to achieve this, every node maintains an append-only ordered set of \emph{finalized blocks}, called a \emph{finalized chain}, where each block contains a number of transactions.
We say that a chain $\chain$ is a \emph{prefix} of chain $\chain'$ if and only if $|\chain| \leq |\chain'|$ and the first $|\chain|$ blocks of chain $\chain'$ are identical to chain $\chain$.
We can then define consistency and liveness with respect to these finalized chains: all honest nodes finalize the same chain of blocks, and when they can communicate with bounded delay, they are able to finalize new blocks.

\begin{definition}[Consistency]
If chain $\chain$ is the finalized chain of an honest node $n_1$ at time $t_1$ and $\chain'$ is the finalized chain of an honest node $n_2$ at time $t_2$, then either $\chain$ is the prefix of $\chain'$, or $\chain'$ is the prefix of $\chain$. 
\end{definition}

\noindent Informally, consistency means that honest nodes' finalized chains never diverge.

\begin{definition}[Liveness]
There exists a polynomial function $f$ such that for any number of protocol participants, in any period of synchrony that is at least $f(N)$ long, each honest node will add at least one new block to its finalized~chain. 
\end{definition}

\noindent Informally, liveness means that when honest nodes can communicate with each other with bounded delay, their finalized chains grow.

\subsection{Threat Model}

Please note that our threat model, that is, an adversary's goals and capabilities, are implicitly defined by the assumptions and objectives that we introduced for honest nodes in the preceding subsections.
Here, we provide further discussion of these from the perspective of an adversary that controls a set of dishonest nodes.
A dishonest node is a node which at any point of time might choose to not follow the protocol. In context to dishonest nodes we will talk about two cases. First, we will talk about  the events in which a dishonest node cannot prevent the correct functioning of the protocol; second,  the events in which a dishonest might alter the state of honest nodes selectively to prevent the correct functioning of the protocol.

First, a dishonest node cannot stop a honest proposer from proposing blocks, and it cannot stop a honest voting node from voting for proposed blocks. Dishonest node cannot change or forge $\Signatures$. During the period of synchrony after a message has been send, a dishonest node cannot prevent the message  from being delivered in at most $\Delta$ time after the message was send. Outside period of synchrony, a dishonest node cannot prevent messages from getting delivered. A dishonest node cannot prevent messages from getting delivered in ordered point to point messages.

Second, a dishonest node acting as a proposer can selectively withheld the votes of the last notarized block in its longest notarized chain selectively from some honest nodes or all honest nodes. A dishonest node, if it has not shared the votes for some blocks in its longest notarized chain, then it can prevent honest nodes from moving to the next epoch for certain duration of time by selectively sending votes for blocks notarized in its longest notarized chain to some honest nodes, so that they do broadcast $\Timeout$ message. A dishonest node can send any message as many times it wants~to.
\section{\Pipelet Data Structures and Messages}
\label{sec:dataStructDescp}

In this section, we introduce  data
structures, notations, and message types used by the \Pipelet protocol.
For a comprehensive list of notation used throughout this paper, please see \cref{tab:notation} in \cref{sec:notation}.

\subsection{Data Structures}

\subsubsection{Block}
A block is a tuple of the form $\langle e, s, \textit{TXs}, h \rangle$ where $e \in \mathbb{N}$ identifies the epoch (i.e., time period) to which the block belongs, $s \in \mathbb{N}$ is a  sequence number (i.e., a unique block identifier within the epoch),  $\textit{TXs}$ is a set of transactions (i.e., transactions recorded by the block), and $h$ is the hash value of a previous block (see below).

\subsubsection{Blockchain} 
A blockchain is an ordered set of blocks, where every block except the first one contains the hash value $h$ of the preceding block in the chain. 
The first block of the chain is called the \emph{genesis block}, which always has the form
$\langle 0, 0, \{\emptyset\}, 0 \rangle$. 
We let $G^*$ denote the genesis block throughout this paper. 

\subsubsection{Notation}
Next, we introduce notations that we use in our protocol specification and proofs.
We define the \emph{length} $|\chain| \in \mathbb{N}$ of chain $\chain$ as the number of blocks in $\chain$ except for the genesis block $G^*$.
Similarly, we define the \emph{length} $|B| \in \mathbb{N}$ of a block $B$ as the length $|\chain]$ of the chain $\chain$ that ends with block $B$.
Note that for any block $B$, there exists exactly one chain $\chain$ that ends with $B$ since every block references a unique preceding block using the hash value $h$.

We let $\chain[l]$ denote the block at length $l$ in $\chain$ where $1 \leq l \leq |\chain|$. 
Then, we let
$\chain[-l] = \chain[|\chain|-l+1]$.
Further, we let $\chain[i:j]$ denote the ordered set of blocks from $\chain[i]$ to $\chain[j]$ in $\chain$ (inclusive of blocks $\chain[i]$ and $\chain[j]$).
Finally, we let $B.e$ and $B.s$ denote the epoch number $e$ and sequence number $s$ of block~$B$.

\subsubsection{Normal and Timeout Blocks}
To define the validity of blockchains, we also introduce the notion of normal and timeout blocks.
First, a block $\chain[i]$ is a \emph{normal block} if $\chain[i].e = \chain[i-1].e$ and $\chain[i].s = \chain[i-1].s + 1$.
Second, a block $\chain[i]$ is a \emph{timeout block} if $\chain[i].e > \chain[i-1].e$ and $\chain[i].s = 1$ .

\subsubsection{Blockchain Validity} 
We say that a blockchain $\chain$ is valid if an only if
\begin{enumerate}
\item for every $1 \leq i \leq |\chain|$, $\chain[i].h = H(\chain[i-1])$;
\item  for every $1 \leq i \leq |\chain|$, block $\chain[i]$ is either a normal or timeout block.
\end{enumerate}
Please recall that $H$ is a collision-resistant hash function.
We also require the transactions recorded by the blockchain to be valid; however, we do not discuss specific rules of validity for transactions  since %
the goal of this paper is to build a permissioned consensus protocol (instead of focusing on a specific application).

\subsection{Messages Types}

\Pipelet nodes reach consensus by %
exchanging four types of messages: \Proposal, \Vote, \Timeout, and \Sync messages.
Please recall that every message is digitally signed and has an implicit sender, which we omit from the notation for the sake of simplicity.

\begin{itemize}
\item {$\Proposal$}: A proposer node $n \in N$ can send this message to propose a new block which may eventually be recorded on the finalized chain. The $\Proposal$ message along with the proposed block also contains $\Signature$ of votes of the previous notarized block. The proposal message has the format: \vivek{latex styling problem} 
$\Proposal\langle B',$ $Votes[B], \Signature \rangle$, where $B'$ represents the proposed block, $Votes[B]$ represent the aggregation of \emph{signatures} on block $B$ voted upon by the voting nodes, and $\Signature$ represents the \emph{signature} of the current proposer on the proposed block $B'$.

\item {$\Vote$}: A voting node $n \in N$ can send this message to vote for (i.e., endorse) a proposed block to the proposer node of the proposed block. The vote message has the format: $\Vote(B, \Signature)$, where $B$ is the block voted upon and \Signature represents the sign of voting node on block $B$. 

\item {$\Timeout$}: A voting node $n \in N$ can send this message to indicate its willingness to advance to the next epoch $e'$. A \vivek{latex styling problem} \Timeout message has the format:

$\Timeout \langle e', \Signature \rangle$, where $e'$ represents the epoch to move to next and \Signature represent the signature by the voting node on the next epoch $e'$.

\item {$\Sync$}: Any node $n \in N$ can send its set of longest notarized chains and the votes belonging to those chains to other nodes using the $\Sync$ message. 
\end{itemize}

\section{\Pipelet Protocol Specification}
\label{sec:protocol}
Here, we provide a complete and formal specification of the  \Pipelet protocol. 
We specify the protocol by describing
the behavior of the participating nodes;
in particular, we specify how a node responds to an event (e.g., respond to receiving a message) by updating its state and/or by sending messages.

In \cref{ns}, we define the state of a node and how the state is initialized. 
In \cref{pd}, we describe how nodes synchronize their epochs, how an eligible proposer node is chosen for an epoch, how blocks are proposed and voted upon, and how nodes finalize blocks.

\subsection{Node State} \label{ns}

\subsubsection{Clock Variables}
Every voting node has clock variable, called $\textit{EpochTimer}$ and the proposer node has a clock variable, called $\textit{NotarizationTimer}$. Both these variables measure real time, and may be reset to zero (i.e., $\textit{EpochTimer}$, $\textit{NotarizationTimer}$ are the amount of real time elapsed since the last reset. 

\subsubsection{Epoch and Sequence Numbers}
The \Pipelet protocol divides time into intervals, called \emph{epochs}, which are identified by integers.
Since the nodes' clocks are not synchronized to each other, 
every node (including both proposer and voting nodes) maintains its own epoch number~$e$. 
Within each epoch, the \Pipelet protocol allows the eligible proposer to propose multiple blocks, which are identified by their sequence numbers. 
To keep track of blocks within an epoch, every node (including both proposer and voting nodes) maintains its own sequence number~$s$.
Both the epoch and sequence numbers are initially set to $1$ (i.e., initially, $\langle e,s \rangle := \langle 1,1 \rangle$ at every node).  

\subsubsection{Received Messages} 
Every node maintains a set \vivek{latex styling problem} 
$\textit{Timeout}$- $\textit{Signatures}(e')$ for every epoch number~$e'$,  which represents the $\Signatures$ to move to epoch $e'$ send by different different voting nodes.
Every node also maintains a set $\textit{Votes}(B)$ for every block~$B$, which represents the set of voting nodes from which the node has received $\Vote(B,\Signature)$ message.
Both sets are initially set to empty sets (i.e., initially,  $\textit{TimeoutSignatures}(e') = \emptyset$ for every $e'$ and $\textit{Votes}(B) = \emptyset$ for every $B$).
Note that in practice, a node can initialize these sets in a ``lazy'' manner, when the first $\textit{TimeoutSignatures}(e')$ or $\Vote(B)$ message is received for an $e'$ or $B$.
Further, once a node has advanced to epoch $e'$ (i.e., once $e \geq e'$), it does not need to maintain set $\textit{TimeoutSignatures}(e')$ anymore. Similarly, once it has finalized a block $B$, it no longer needs to maintain $\textit{Votes}(B)$.
Every node maintains a set $\textit{Notarized}$, which represents the set of notarized blocks the node maintains. One way to add blocks to the $\textit{Notarized}$ set is by receiving $\Sync$ messages, which represent the set of longest notarized chains of the sender node. We talk more about the notarized blocks in \cref{blc1}.

\subsubsection{Blocks} \label{blc1}
To finalize blocks, the \Pipelet protocol uses the notion of notarization.
If a block has received votes from at least two thirds of the voting nodes, the block is considered to be notarized. 
Then, if three consecutive normal blocks are notarized, then the middle block and all preceding blocks on its chain are considered to be finalized; we will present the formal notarization rule later. 

The $\textit{Notarized}$ set contains all the notarized blocks in the node's view.
By ``\emph{in the node's view},'' we mean ``\emph{based on all the messages that the node has received}.''

For ease of presentation, we let $\textit{Longest}$ denote the set of longest notarized blocks for each node.
Formally, for every node, we let $\textit{Longest} = \operatorname{argmax}_{B \in \textit{Notarized}} |B|$.

Finally, every node maintains a set $\textit{Finalized}$, which represents the set of finalized blocks in the node's view.
Note that maintaining $\textit{Finalized}$ is redundant since the set of finalized block can always be determined based on the set of notarized blocks $\textit{Notarized}$ and on the finalization rule.
Nonetheless, we maintain the set $\textit{Finalized}$ for ease of presentation and since this is the main output of the protocol.
Both the set of notarized block and the set of finalized blocks initially contain only the genesis block (i.e., initially, $\textit{Notarized} = \{G^*\}$ and $\textit{Finalized} = \{G^*\}$).

\begin{figure}[t]
    \centering
    \includegraphics[width=\columnwidth]{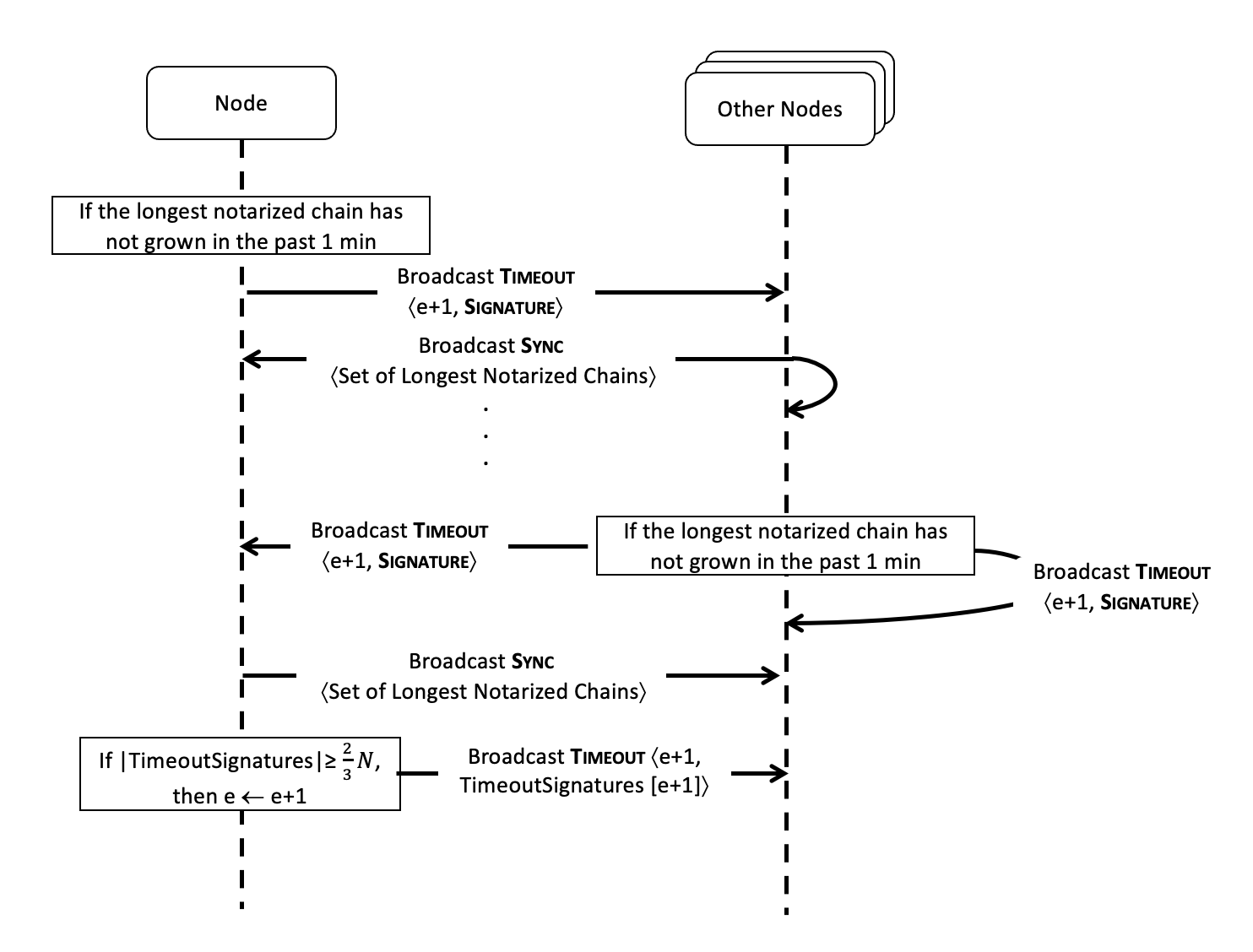}
    \caption{Process of clock and block synchronization.} 
    \label{fig:seq_clock}
\end{figure}

\subsection{Protocol Description} \label{pd}

\subsubsection{Clock Synchronization} \label{clock_des} 

First, we describe how a node advances its local epoch number $e$.
Whenever a node (either voting node or proposer node) receives $\Signatures$ from at least two thirds of the voting nodes for some epoch $e' > e$, indicating the nodes' willingness to advance to epoch $e'$, then the node immediately advances its own epoch number to $e'$ (see \cref{alg:sharedF1}).
Formally, if \vivek{latex styling problem} $\textit{TimeoutSignatures}(e')$ $\geq \frac{2}{3} |N|$ for some $e' > e$, then $e \gets e'$.

Second, we specify when voting nodes send their \vivek{latex styling problem} 
$\Timeout($ $e', \Signature)$ messages.
The purpose of these messages is to advance to the next epoch, and hence, switch to the next eligible proposer if the longest notarized chain has not grown in a while (i.e., if the protocol is ``stuck'').
Voting nodes achieve this using their $\textit{NotarizationTimer}$: when advancing to a new epoch, voting nodes always reset their $\textit{NotarizationTimer}$ and $\textit{EpochTimer}$ to $0$; similarly, when extending a longest notarized chain with a new notarized block, voting nodes reset their $\textit{NotarizationTimer}$ to $0$.
If the $\textit{NotarizationTimer}$ of a voting node ever reaches $1 min$, it immediately broadcasts a $\Timeout((e+1), \Signature)$ message to every other node (see \cref{alg:vote}).
In other words, a voting node sends a \Timeout message, indicating its willingness to advance to the next epoch, if its longest notarized chain has not grown in $1 min$ (and it has already spent $1 min$ in the current epoch).
Note that sending the $\Timeout((e+1), \Signatures)$ message does not advance the node to the next epoch; it needs to see $\Signatures$ from two thirds of voting nodes to advance, just like all the other nodes. And as soon as a voting node enters next epoch $(e+1)$ it immediately broadcast $\Timeout\langle (e+1), \textit{TimeoutSignatures}[e+1]  \rangle $.

\begin{algorithm}[th]
\caption{Clock Synchronization and Finalization}\label{alg:sharedF1}
\begin{algorithmic}[1]

\State \textbf{Event:} On reception of message $\Timeout\langle e',$ $\Signatures \rangle$ from node $n'$
\For{block in $\langle \textit{Longest}~\setminus~\textit{SyncedBlocks} \rangle$}
\State {$\textit{SyncVotes[block]} \gets \textit{SyncVotes[block]} \cup \textit{Votes[block]}$}
\EndFor
\State Broadcast $\Sync \langle \langle \textit{Longest}~\setminus~\textit{SyncedBlocks} \rangle, \langle \textit{SyncVotes} \rangle \rangle$
\State {$\textit{SyncedBlocks} \gets \textit{SyncedBlocks}~\cup~\textit{Longest}$}
\If{$e' > e$} 
\State $\textbf{aggregate:}~\textit{TimeoutSignatures}[e'] \gets$\newline$~~~~~~~~~~~~~~~~~~~~~~~~\textit{TimeoutSignatures}[e'] \cup \Signatures$
\vspace{0.15em}
\If {$\textbf{verify}\left(|\textit{TimeoutSignatures}[e']| \geq \frac{2}{3} |N|\right)$}
  \State $e \gets e'$
  \State $s \gets 1$
  \State $\textit{EpochTimer} \gets 0$
  \State $\textit{NotarizationTimer} \gets 0$
  \State {Broadcast $\Timeout \langle e', \textit{TimeoutSignatures}[e'] \rangle$}
  \EndIf
\EndIf

\\

\State \textbf{Event:} On notarization of block $B$
\If {$\exists \{B_{1}, B_{2}\} \in \textit{Notarized}$ \textbf{such that} $(B.h = H(B_{2})  \land B_{2}.h = H(B_{1})) \land B_{2} \notin \textit{Finalized} \land B_{1}.s > 1$}
\State  {$\textit{Finalized} \gets \textit{Finalized} \cup \{B_{2}\}$}
\EndIf
\end{algorithmic}
\end{algorithm}

\subsubsection{Notarizations}
A node's longest notarized chain can grow in two ways: 
by receiving at least $\frac{2}{3}|N|$ votes on the proposed block or 
by receiving a $\Sync$ message. A honest node upon receiving a \Timeout message, immediately broadcast the a \Sync message which includes the longest notarized chains of the sending node. If the longest notarized chain of a honest node grows, the $\textit{NotarizationTimer}$ is reset to zero. For  a pseudo-code of the Notarization rule, please see \cref{alg:proposer} and \cref{alg:vote}.

\subsubsection{Finalization}
 When a node (either voting or proposer node) encounters three consecutive normal blocks with epoch and sequence numbers $(e,s), (e, s+1), (e, s+2)$ in a fully notarized chain the node has observed, the node finalizes the middle block of the three consecutive blocks, as well as the entire prefix chain of the middle block. For a node, finalization makes all transaction in the finalized block and in its prefix chain permanent. 
 
For a pseudo-code of the clock functionality and finalization rule, please see \cref{alg:sharedF1}.
For an illustration, see \cref{fig:seq_clock}.

\begin{figure}[t]
    \centering
    \includegraphics[width=\columnwidth]{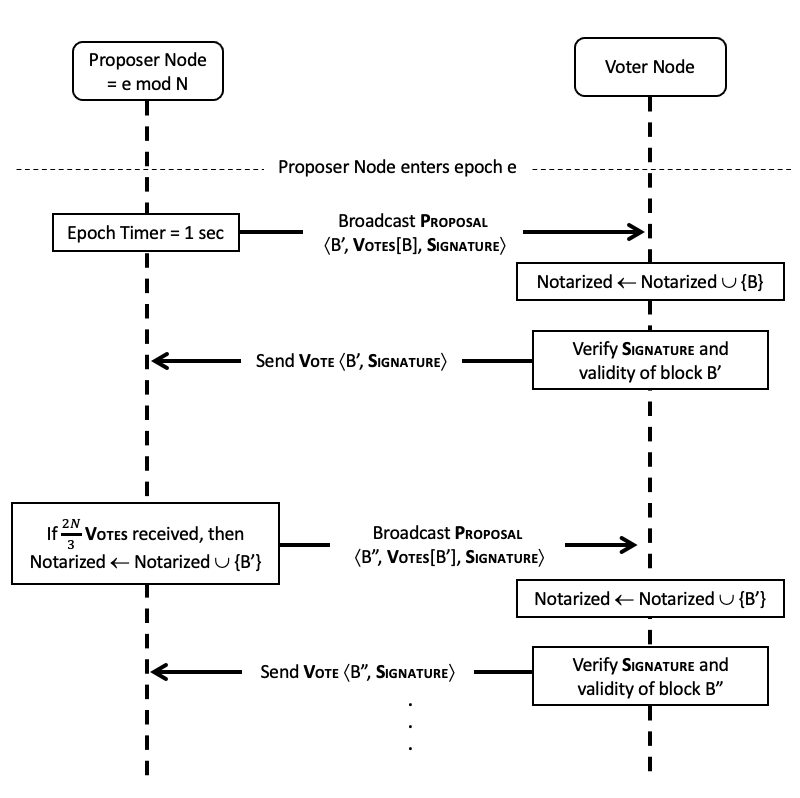}
    \caption{Process of block proposal and voting.}
    \label{fig:notarization}
\end{figure}

\subsubsection{Proposal}\label{prop_desc}
For each epoch, we designate one proposer node to be eligible to  propose blocks for that epoch. We consider a \emph{stable proposer} approach, where stable proposer means the ability to propose more than one block per epoch. To choose a proposer for an epoch, we use a simple round robin approach: for epoch $e$, the $(e \operatorname{mod} N)$th proposer is considered to be the eligible proposer (i.e., we number proposers from $1$ to $|N|$, and let the proposer whose number is $(e \operatorname{mod} N)$ be eligible for epoch $e$). In \cref{alg:proposer}, \emph{identifier of node} represents the unique proposer identifier number, which ranges from $1$ to $|N|$. Whenever a proposer node enters a new epoch $e$, it first waits for  $1sec$, so that the current proposer can receive a longest notarized block belonging to the epoch $\langle e - 1 \rangle$ in case the current proposer's longest notarized block belongs to a epoch earlier than epoch $\langle e - 1\rangle$, and if the node is an eligible proposer for epoch $e$, then the node performs the following  from that point on:

\begin{enumerate}
\item First, let $\chain_L$ be one of the longest fully notarized chains that the proposer node has observed so far. Then, propose the timeout block $B:= \langle  e,1, \textit{TXs}, H(\chain_{L}[-1])\rangle$ extending from $\chain_L$, where $\textit{TXs}$  is a set of outstanding transactions, and broadcast $B$ to all the nodes.  
\item Next, repeat the following: if the node has proposed blocks at sequence numbers $\langle e, 1 \rangle, \ldots, \langle e, s \rangle$ so far, and all $\langle e, 1 \rangle, \ldots, \langle e, s \rangle$ have been notarized in the proposer's view, then propose the next normal block containing the outstanding transactions (extending from the chain ending at $\langle e, s \rangle$ that it has proposed so far) and broadcast the proposed block $B:= \langle  e,s+1, \textit{TXs}, H(\chain_{L}[-1])\rangle$ to all the nodes.
\end{enumerate}

The proposal message from the proposer along with the proposed block also contains the \emph{votes} for the last notarized block from the parent chain of the proposed block and also \emph{sign} of the proposer on the current proposed block. 

For a pseudo-code of the proposer functionality, please see \cref{alg:proposer}.
For an illustration, see \cref{fig:notarization}.

\begin{algorithm}[th]
\caption{Proposer Functionality}\label{alg:proposer}
\begin{algorithmic}[1]
\State \textbf{Event:} When $\textit{EpochTimer} = 1sec$
\State \textbf{Do Propose}
\\
\State \textbf{Event:} On reception of message $\Vote\langle B, \Signature\rangle$ from a node $n$
\State $\textbf{aggregate:}~\textit{Votes}[B] \gets \textit{Votes}[B] \cup {\Signature}$
\If {$\textbf{verify}\left(|\textit{Votes}[B]| \geq \frac{2}{3}|N|\right)$}
\State $\textit{Notarized} \gets \textit{Notarized} \cup \{B\}$

\If{$(\textit{EpochTimer} > 1sec) \land (e = B.e)$}
\State \textbf{Do Propose}
\EndIf 
\EndIf
\\
\State \textbf{Function Propose:}
\If{identifier of node $n = e \operatorname{mod} |N|$}
\State {$B \gets \textnormal{arbitrary block} \in \textit{Longest}$}
\State {$B' \gets  \langle e, s, \textit{TXs}, H(B) \rangle$}
\State {$\Signature \gets \textbf{sign}(B') $}
\State {Broadcast $\Proposal\langle B', \textit{Votes}[B], \Signature \rangle$} %
\State {$\textit{Votes}[B'] \gets \{ \Signature \}$}
\State {$s \gets s+1$}
\EndIf
\end{algorithmic}
\end{algorithm}

\subsubsection{Vote} Upon receiving the proposed block of the form $B := \langle e, s', \textit{TXs}, h' \rangle$ from an eligible proposer for epoch $e$, a voting node $n$ votes on the proposed block $B$ extending from a chain $\chain_{L}$(one of the longest notarized chains that the node has observed) and sends back the vote to the proposer of the block being voted upon, if and only if the following~hold:
\begin{enumerate}
    \item the node is currently in local epoch $e$ and it has not voted for any block with the epoch and sequence number pair $\langle e,s' \rangle$.
    \item all blocks preceding $\langle e,s' \rangle$ have already been notarized. %
\end{enumerate}

For a pseudo-code of the voting functionality, please see \cref{alg:vote}.
Note that this pseudo-code also includes the timeout functionality, since it provides a complete description of the voting-node specific functionality of \Pipelet.
For an illustration, see \cref{fig:notarization}.

\begin{algorithm}[th]
\caption{Voter Functionality}\label{alg:vote}
\begin{algorithmic}[1]
\State \textbf{Event:} When $\textit{NotarizationTimer} = 1min$
\State{$\Signature \gets \textbf{sign}(e + 1)$}
\State Broadcast $\Timeout \langle (e + 1), {\Signature} \rangle$
\State $\textit{NotarizationTimer} \gets 0$
\\

\State \textbf{Event:} On reception of $\Sync\langle \langle \textit{Chains} \rangle, \langle \textit{SyncVotes} \rangle\rangle$
\If{$\textbf{verify}(\textit{Chains})$}
\State {$B_{\textit{prev}} \gets \textnormal{arbitrary block} \in \textit{Longest}$}
\State $\textit{Notarized} \gets \textit{Notarized} \cup \textit{Chains}$
\State {$B_{\textit{next}} \gets \textnormal{arbitrary block} \in \textit{Longest}$}
\If {$|B_{\textit{next}}| > |B_{\textit{prev}}|$}
\State $\textit{NotarizationTimer} \gets 0$
\EndIf
\For{{$B$ in \textit{SyncVotes}}}
\State {$\textit{Votes}[B] \gets \textit{Votes}[B] \cup {SyncVotes[B]}$}
\EndFor
\EndIf

\\
\State \textbf{Event:} On reception of message $\Proposal \langle B', \textit{Votes}[B],$ $\Signature \rangle$, where $B' = \langle e', s', \textit{TXs}, h' \rangle$ from node~$n'$
\If {$\textbf{verify}\left(|\textit{Votes}[B]| \geq \frac{2}{3}|N|\right)$}
\State {$B_{\textit{prev}} \gets \textnormal{arbitrary block} \in \textit{Longest}$}
\State $\textit{Notarized} \gets \textit{Notarized} \cup \{B\}$
\If{ $|B| > |B_{\textit{prev}}|$}
\State {$\textit{NotarizationTimer} \gets 0$}
\EndIf
\If{$(e' = e) \land (s' \geq s) \land (n' = e \operatorname{mod} |N_{p}|)$\newline
$~~~~~~~~\land \textbf{verify}(\Signature, B') \land \exists B \in \textit{Longest} ~ \Big[$\newline
$~~~~~~~~~~~~~~\left(h' = H(B)\right) \land \left(|B'| = |B| + 1\right)$\newline
$~~~~~~~~~~~~~~\land \big[ ((e' > B.e) \land (s' = 1))$\newline 
$~~~~~~~~~~~~~~~~~~~\lor ((e' = B.e) \land (s' = B.s + 1)) \big]\Big]$}
\State {$\Signature \gets \textbf{sign}(B')$}
\State {Send $\Vote \langle B', \Signature \rangle$ to $n'$}
\State {$s \gets s' + 1$}
\EndIf
\EndIf
\end{algorithmic}
\end{algorithm}

\section{Protocol Guarantees}
\label{sec:analysis}

In this section, we state our main theorems about the consistency and liveness guarantees of our protocol. 

\subsection{Consistency}
\label{sec:consistency}
Here, we introduce the main theorem for consistency.
\begin{theorem} \label{T13} 
Let us assume the following: (a) more than $\frac{2}{3} |N|$ nodes are honest, (b) $\chain'$ is the finalized chain of an honest node, and $\chain''$ is the finalized chain of another honest node. Then, it must hold that either $\chain'$ and $\chain''$ are the exact same chain, or one of them is a prefix of the other.
\end{theorem}

The \Pipelet protocol guarantees consistency at all times, even outside of periods of synchrony. The above theorem simply states that at any point of time, if we have two finalized blockchains, they must either be the same chains or if not the same chains then one must be the prefix of the other. 

The lemmas---along with their proofs---required to prove the main consistency theorem as well as the proof of the main consistency theorem (\cref{T13}) are provided in \cref{pt13}.

\subsection{Liveness}
\label{sec:liveness}
Here, we introduce the main theorem for liveness.
\begin{theorem} \label{L23L}
Let us assume the following: (a) more than $\frac{2}{3} |N|$ nodes are honest, (b) time interval $[t_{0}, t_{1}]$ is a period of synchrony that is at least $(1sec+7\Delta+1min) \cdot f)+ 4 \cdot (4min + 2\Delta)+ ( 2sec+ 18\Delta+ 1min)$ long, where $f$ is the number of dishonest nodes.
Then, for any honest node, the node's finalized chain at time $t_1$ must be longer than its finalized chain at time $t_0$.
\end{theorem}

Liveness is guaranteed during period of synchrony. We consider $f < \frac{1}{3} |N|$ dishonest nodes, that is, at most $f$ nodes can be corrupt proposers and voters controlled by an adversary. So, in the above theorem, we define the minimum length of period of synchrony that is long enough to guarantee that at least four consecutive blocks (first being the timeout block and the other three being the normal blocks) will be proposed and voted upon. When this happens, as the above theorem shows, at least one new block (that was not finalized before) will be finalized before the period of synchrony ends. In other words, the finalized chain of each honest node at time $t_1$ will be longer than its finalized chain at time $t_0$.

The lemmas---along with their proofs---required to prove the main liveness theorem as well as the proof of the main liveness theorem (\cref{{L23L}}) are provided in \cref{LP}.
\section{Performance} 
\label{perf} \vivek{in caption for figures i have mentioned DELAY(scale for exponential distribution) as 0.5 and 2. not sure, in figures how else to mention SCALE for exponential delay. Can you please update it if you have a better way in mind  } 

\begin{figure}[t] 
    \centering
    \includegraphics[width=\columnwidth]{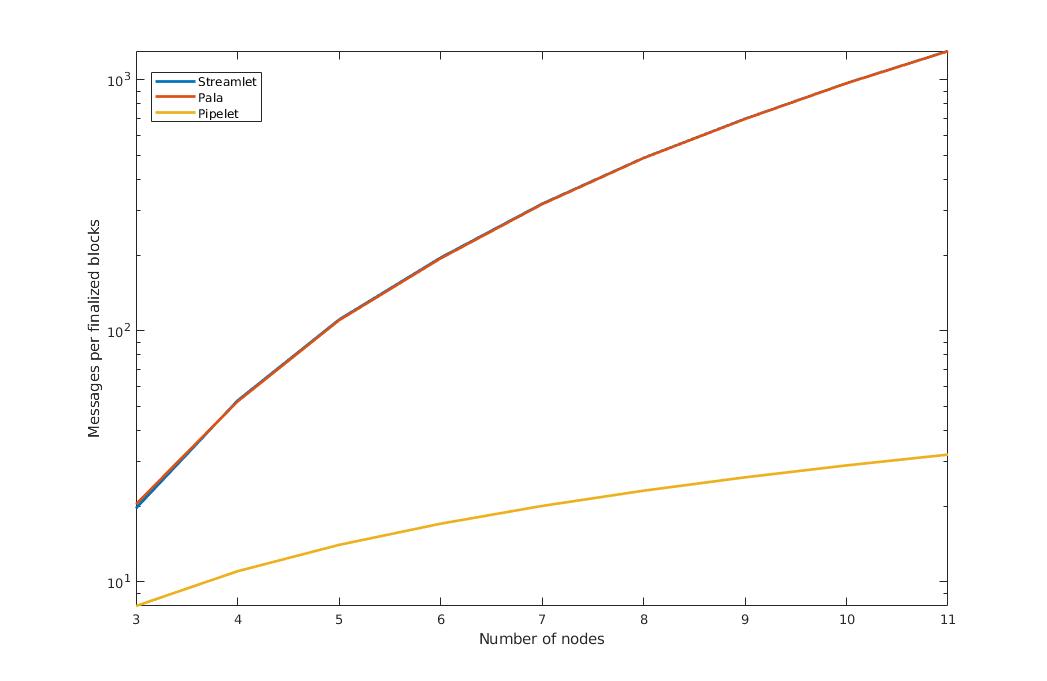}
    \caption{Messages per finalized blocks for average delay of 0.5.} 
    \label{fig:mess_0.5}   
\end{figure}

\begin{figure}[t]
    \centering
    \includegraphics[width=\columnwidth]{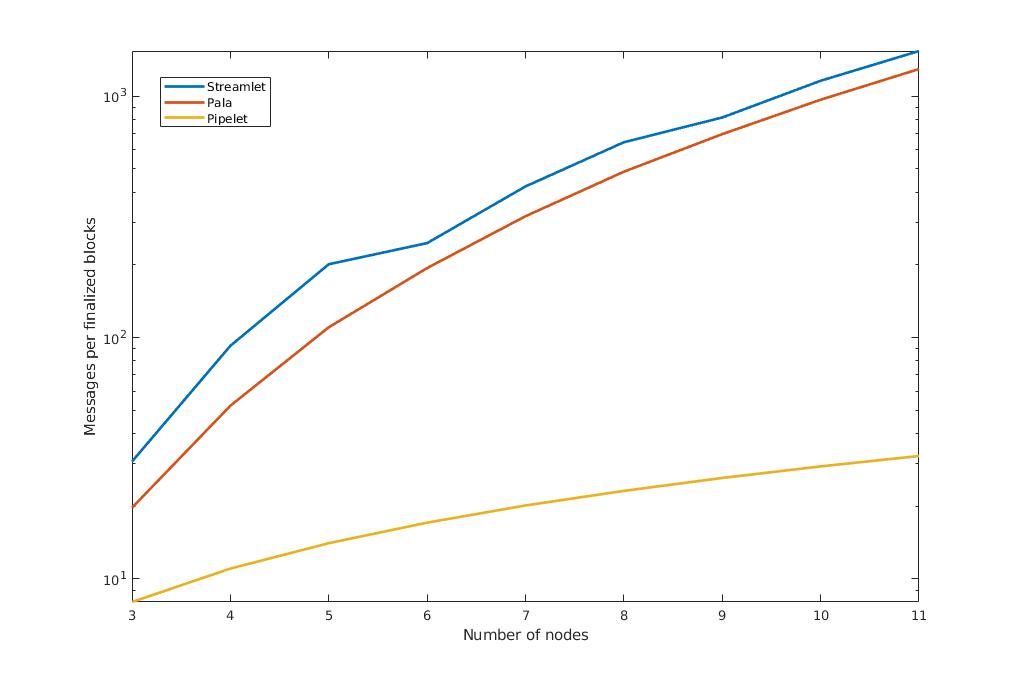}
    \caption{Messages per finalized blocks for average delay of 2.} %
    \label{fig:mess_2}
\end{figure}

\begin{figure}[t]
    \centering
    \includegraphics[width=\columnwidth]{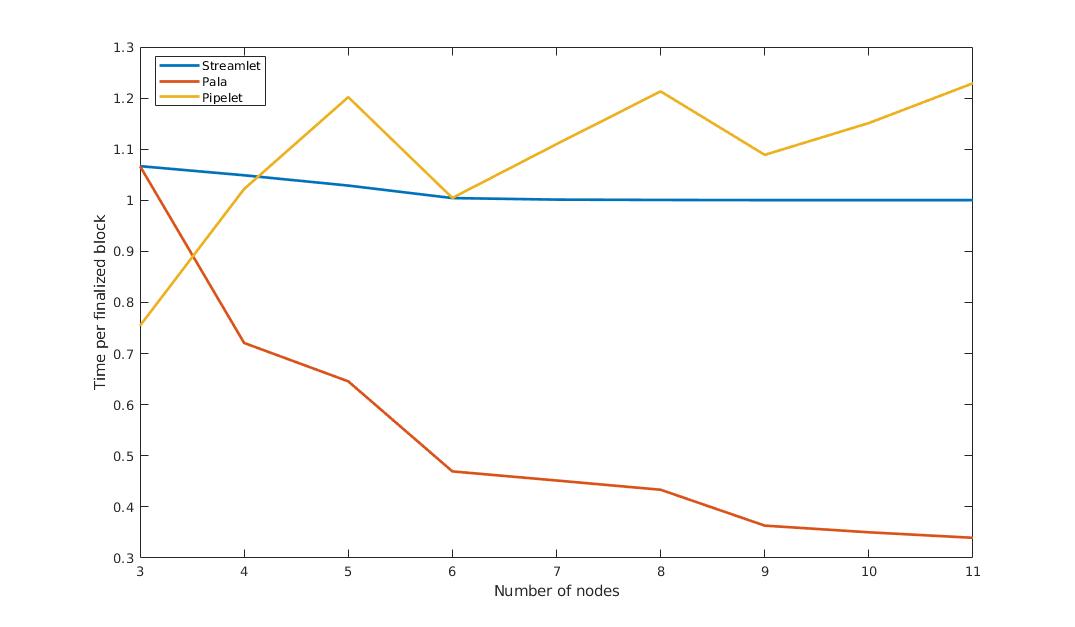}
    \caption{Time per finalized blocks for average delay of 0.5.} %
    \label{fig:time_0.5}
\end{figure}

\begin{figure}[t]
    \centering
    \includegraphics[width=\columnwidth]{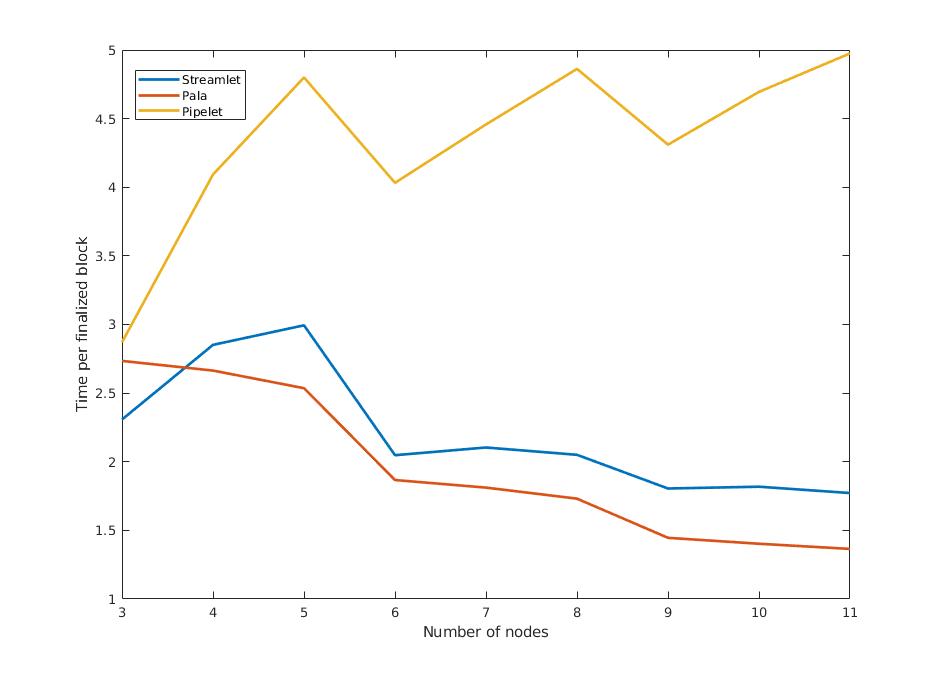}
    \caption{Time per finalized blocks for average delay of 2.} %
    \label{fig:time_2}
\end{figure}

We have simulated Streamlet~\cite{chan2020streamlet}, PaLa~\cite{chan2018pala}, Pipelet using Python language outside period of synchrony. To provide real world delays we use exponential distribution in message passing in the form of \href{https://numpy.org/doc/stable/reference/random/generated/numpy.random.exponential.html}{numpy.random.exponential}. 
In \cref{fig:mess_0.5} and \cref{fig:mess_2} we have plotted total number of messages send in the simulation per finalized block vs number of nodes. The average delay in message transmission per message is more for experiment conducted for results in \cref{fig:mess_2} as compared to experiment conducted in \cref{fig:mess_0.5}. In \cref{fig:mess_0.5}, as the number of nodes increases, we see number of messages send per finalized block is almost similar for Streamlet~\cite{chan2020streamlet} and PaLa~\cite{chan2018pala} while in \cref{fig:mess_2} PaLa~\cite{chan2018pala} requires to send less number of messages per finalized block. This is because PaLa~\cite{chan2018pala} a proposer can propose multiple blocks per epoch while in Streamlet~\cite{chan2020streamlet} we have the proposer can only propose one block per epoch. Moreover, both in \cref{fig:mess_0.5} and \cref{fig:mess_2}, we see Pipelet requires messages to be transmitted per finalized block as compared to both Streamlet~\cite{chan2020streamlet}, PaLa~\cite{chan2018pala}. This is because both Streamlet~\cite{chan2020streamlet}, PaLa~\cite{chan2018pala} require implicit echoing of messages, while in Pipelet there is no implicit echoing required.

In \cref{fig:time_0.5} and \cref{fig:time_2} we have plotted total time spend by all nodes per finalized block in the simulation versus number of nodes. The average delay in message transmission per message is more for experiment conducted for results in \cref{fig:time_2} as compared to experiment conducted in \cref{fig:time_0.5}. Both in \cref{fig:time_0.5} and \cref{fig:time_2}, as the number of nodes increases, we see time spend by nodes in simulation per finalized block for Pipelet is more than both PaLa~\cite{chan2018pala} and Streamlet~\cite{chan2020streamlet}. It is because implicit echoing leads to more blocks proposed in both PaLa~\cite{chan2018pala} and Streamlet~\cite{chan2020streamlet} as compared to Pipelet.
Also, for Pipelet both in \cref{fig:time_0.5} and \cref{fig:time_2} we see a up and down zig zag result as the number of nodes increase. This is because of the reason, if the number of nodes in the simulation is not a perfect multiple of three, then the percentage of nodes required to notarize a block is more as compared to when the total number of nodes is a perfect multiple of three. Hence the up and down zig zag line plot for Pipelet.

\section{Conclusion}
\label{sec:concl}

We described a simple and performant BFT-style consensus protocol, called \Pipelet, that is based on familiar rules such as extending longest chain, and finalizing the  middle of three consecutive normal notarized blocks, and stable leader for performance. 
The \Pipelet protocol is conceptually different to Streamlet~\cite{chan2020streamlet} and PaLa~\cite{chan2018pala}, in the sense \Pipelet does not have implicit echoing which significantly reduces the message cost required to finalize a block in normal scenario to $O(N)$ instead of $O(N^3)$ for Streamlet~\cite{chan2020streamlet} and PaLa~\cite{chan2018pala}. Moreover, \Pipelet combines the of advantages of Streamlet and Pala in terms of simplicity, performance, and practicality. In addition to it, from the experimental results in \ref{perf}, we conclude that both Streamlet~\cite{chan2020streamlet} and PaLa~\cite{chan2018pala} messages increase exponentially with the increase number of nodes. So, as compared to Pipelet both Streamlet~\cite{chan2020streamlet} and PaLa~\cite{chan2018pala} are impractical for networks with large number of nodes.
We also provided detailed specification of the assumptions and the protocol, and proved its consistency and liveness under partial synchrony.

 \bibliographystyle{ACM-Reference-Format}
\bibliography{main}
\begin{appendix}

\section{Proof of Consistency} 
\label{CP}

For the the rest of the paper, we assume $B.es$ denote the pair $\langle e, s  \rangle$ for a block~$B$.

\begin{lemma}[Uniqueness of each sequence number]\label{L12}
Let us assume the following (a) more than $\frac{2}{3} |N|$ nodes are honest, (b) $\chain'$ and $\chain''$ be two chains notarized by the voting nodes belonging to $N$, observed by all honest nodes such that $\chain'[-1].es = \chain''[-1].es$, it must be that $\chain'$ = $\chain''$. Please recall that the $es$ suffix in $\chain'.es$ and $\chain''.es$ represent their respective epoch and sequence number combination of that block. 
\end{lemma}

The lemma statement is same as {Lemma 5} of PaLa~\cite{chan2018pala}.

\begin{proof}
This lemma has the exact same proof as Lemma 5 of the PaLa~\cite{chan2018pala} paper because we employ the same constraint of having more than $\frac{2}{3} |N|$ honest nodes, and the proof is unaffected by the differences in the design of the \Pipelet protocol. 
\end{proof}

Based on the preceding lemmas, we now can prove our main consistency theorem.

\begin{proof}[Proof of \cref{T13}] 
\label{pt13}
In the proof below $\chain'$ represent the longest notarized chain out of which the finalized chain $\chain_1$ was obtained and similarly $\chain''$ represent the longest notarized chain out of which the finalized chain $\chain_2$ was obtained.

For contradictions sake, suppose the finalized version of $\chain'$ and finalized version of $\chain''$ are not prefix of each other and let $B_0$ be the last common block of the finalized chains with $i_0$ be the index such that $B_{0} = \chain'[i_{0}] = \chain''[i_{0}]$ which has epoch number $e_{0}$. For each $\chain'$ and $\chain''$, we define  $PivotBeginIndex$ block and $PivotEndIndex$ block as follows. 

If all three blocks in $\chain'[i_{0} : i_{0} + 2]$ have the same epoch number and if $\chain'[i_{0}].s \neq 1$, then the $PivotEndIndex$ block $B_2$ in $\chain'$ is $\chain'[i_{0} + 2]$ and the $PivotBeginIndex$ block $B_1$ is $\chain'[i_{0} + 1]$. If $\chain'[i_{0}].s := 1$ and all four blocks in $\chain'[i_{0} : i_{0} + 3]$ have the same epoch number, then the $PivotEndIndex$ block $B_2$ in $\chain'$ is $\chain'[i_{0} + 3]$  and the $PivotBeginIndex$ block $B_1$ is $\chain'[i_{0} + 1]$. Otherwise, suppose $e$ is the smallest epoch in $\chain'$ strictly larger than $e_0$ such that there is a block in chain with sequence number $\langle e, 4 \rangle$ then, in this case, the $PivotEndIndex$ block $B_2$ in chain is the block with sequence number $\langle e, 4 \rangle$ and the $PivotBeginIndex$ block $B_1$ is $\langle e, 1 \rangle$.

Observe that $B_1$ and $B_2$ have the same epoch number, and they are both after $B_0$ in $\chain'$. We define the blocks $B_{1}'$ and $B_{2}'$ in the same manner for $\chain''$. 

To prove the \cref{T13}, we use contradiction in the below cases. Please note both $\chain'[i_{0} + 1]$ and $\chain''[i_{0} + 1]$ can not be normal blocks, because then they have the same epoch, sequence number combination. And then that will contradict Lemma \ref{L12}. Hence, for the rest of the proof, we can assume that at least one of $\chain'[i_{0} + 1]$ and $\chain''[i_{0} + 1]$ is a timeout block.

Lets us define $\langle e_{0}, s_{0} \rangle  := B_{0}.es$, $\langle e, s \rangle  := B_{2}.es$ and $\langle e', s' \rangle  := B_{2}'.es$ without loss of generality and assume $ e \leq e'$

    \textbf{Case 1.} 
            $e = e'$ : The current case $e_{0} = e = e'$ leads to both $\chain[i_{0} + 1]$ and $\chain'[i_{0} + 1]$ being normal blocks, which was already shown to be impossible. Hence, the case $e_{0} < e = e'$ implies that both chains have timeout blocks $B$ and $B'$ after $B_{0}$ with sequence number $(e, 1)$ and this contradicts Lemma~\ref{L12}.

    \textbf{Case 2.} $e < e'$: We divide it further in two cases: (1) $e_{0} = e$, and (2) $e_{0} < e$.
    \begin{enumerate}
        \item \label{T13_a} Consider the case $e_{0} = e$. It must be the case that in $\chain'$, the normal blocks $B_{1} = \chain'[i_{0} + 1]$ and $B_{2} = \chain'[i_{0} + 2]$ have the same epoch number $\langle e_{0} , s_{0} + 1 \rangle $ and $\langle e_{0} , s_{0} + 2 \rangle $ respectively. 
        Furthermore, since $B_{2}$ is notarized by the nodes belonging to the voting set $N_v$, there must have been a set $S \subseteq N_{v}$ of more than $\frac{1}{3} |N_v|$ honest nodes, each of which must have notarized $B_2$ in its local epoch $e = e_{0}$ . This means that during its local epoch $e_0$, such an honest node in $S$ must have seen the full notarization of the prefix of chain up to $B_1$. 
        
        Now the block $B'_{*}$ $= \chain''[i_{0}+1]$ must be a timeout block and suppose it has a epoch, sequence number combination $\langle e'',1 \rangle$ for some epoch $e''>e_0$. 
        
       Since epoch $e'' > e$ and epochs do not decrease with time, so by the beginning of $\langle e'',1 \rangle $, every honest node in $S$ would have observed the block $B_1$ which is at a length $|B_1| > |B_0|$. Thus every honest node in $S$ would have observed a notarized chain ending at $B_1$ which is longer than the notarized chain ending at $B_0$. Thus less than $\frac{2}{3} |N_v|$ vote for block $B'_{*}$
        
        \item   Consider the case $e_{0} < e$. In this case, the stabilized block $B_1$ is a timeout block and will have epoch, sequence number $\langle e, 1 \rangle $ and $B_2$ has sequence number $\langle e,4 \rangle $. Since $B_2$ is notarized by the nodes belonging to the voting set $N_v$, there must be a subset $S \subseteq N_v$ of more than $\frac{1}{3} |N_v|$ honest nodes, each of which has seen the fully notarized prefix of $\chain'$ up to $B_1$ during its local epoch $e$.
        
        Observe that $e < e'$ implies that $\chain''$ must have some timeout block from $\chain''[i_{0} + 1]$ to $B_{1}'$ whose epoch number is strictly larger than $e$. Hence, consider the timeout block $B_{*}'$ in $\chain''$  with sequence number $\langle  e'' , 1 \rangle $ such that $e''$ is the minimum epoch larger than $e$.  Now it can have to sub~cases:
        \begin{enumerate}
            \item $B'_{*}$ is the same block as $\chain''[i_{0}+1]$: Contradiction can be proved in this case using the same logic as above Case 2. point-(\ref{T13_a}) 
            \item $B'_{*}$ is not the same block as $\chain''[i_{0}+1]$: Let $B''_{*}$ represent the block $\chain''[|B'_{*}|-1]$. That is block $B''_{*}$ represent the previous block to $B'_{*}$ in $\chain''$. Let us assume block $B''_{*}$ has epoch number $e'''$ and because block $B'_{*}$ which is the immediate next block to block $B''_{*}$ is a timeout block, thus $e''' < e''$. Now next need to prove $e''' \neq e$. Now, the timeout block $B_1$ with epoch $e$ belongs to $\chain'$. All the normal block belonging to epoch $e$ must also belong to $\chain'$. This is because when every honest node in $N_c$ receives the normal block proposal such that the normal block belongs to epoch $e$, every honest node must have already observed $\chain'$ as one of the longest notarized chain. Now, since block $B''_{*}$ with epoch $e'''$ does not belong to $\chain'$, $e''' \neq e''$. So, since we already know $e''$ is the minimum epoch of a block in $\chain''$ larger than $e$, and $(e'''<e'') \land (e''' \neq e)$, it leads to the conclusion $e'''<e$. Let $B_k$ represent the block $\chain'[|B_{1}|-1]$, that is $B_k$ represents the previous block of $B_1$ in $\chain'$. Consider the following three cases:
            \begin{enumerate}
                \item \label{equal_case}  Length of prefix chain of $B_1$ is equal to length of prefix chain of $B'_*$, that is $|B_{k}|$ is equal to $|B''_{*}|$: This implies, $|B_1|>|B''_{*}|$. Since $B_{2}$ is notarized by the voting set $N_v$, there must have been a set $S \subseteq N_{c}$ of more than $\frac{1}{3} |N_v|$ honest nodes, each of which must have notarized $B_1$ in its local epoch $\langle e,1 \rangle$ . This means that during its local epoch $\langle e,1 \rangle$, such an honest node in $S$ must have seen the full notarization of the prefix of chain up to $B_1$. 
        
               Since epoch $e'' > e$ and epochs do not decrease with time, so by the beginning of $\langle e'',1 \rangle $, every honest node in $S$ would have observed a fully notarized prefix of chain up to $B_1$ which is longer than the fully notarized prefix of chain up to $B''_{*}$. Therefore less than $\frac{2}{3} |N_v|$ will vote for $B'_{*}$. Hence the contradiction.

                \item  Length of prefix chain of $B_1$ is less than the length of prefix chain of $B'_*$, that is $|B_{k}| < |B''_{*}|$: Furthermore, since $B'_{*}$ is notarized by the voting set $N_v$, there must have been a set $S \subseteq N_{v}$ of more than $\frac{1}{3} |N_v|$ honest nodes, each of which must have notarized $B''_{*}$, such an honest node in $S$ must have seen the full notarization of the of chain $\chain'$ up to~$B''_{*}$. 
        
               Since epoch $e > e'''$ and epochs do not decrease with time, by the beginning of $\langle e,1 \rangle $, every honest node in $S$ would have observed a fully notarized prefix of chain up to $B''_{*}$ which is longer than the fully notarized prefix of chain up to $B_k$. Therefore less than $\frac{2}{3} |N_v|$ will vote for $B_1$. Hence the contradiction

                \item  Length of prefix chain of $B_1$ is greater than the length of prefix chain of $B'_*$, that is, $|B_{k}| > |B''_{*}|$: 
                We already proved contradiction for the case when $|B_{k}|$ is equal to $|B''_{*}|$, see Case 2.\ref{equal_case}. Using a similar logic, contradiction can be proved in the current case.
            \end{enumerate}        
    \end{enumerate}    
    \end{enumerate}
\end{proof}

\section{Proof of Liveness} 
\label{LP}

\begin{lemma} \label{F1}
Let us assume the following (a) more than $\frac{2}{3} |N|$ nodes are honest, (b) let $[t_{0}, t_{1}]$ denote a period of synchrony. It must be that if some honest node $n$ is in local epoch $e$ at some time $r \in [t_{0} - \Delta, t_{1}-\Delta]$, then by time $r+\Delta$, no honest nodes is in epoch $e'$ such that $e' < e$. 
\end{lemma}

\begin{proof}
We know a honest node $n$ is in epoch $e$ at time $r$. Then, we can conclude that at most by time $r$, node $n$ must have entered epoch $e$. For a honest node to enter a epoch it requires $\frac{2}{3}|N|$ $\Signatures$ for that particular epoch. Hence node $n$ must have received for epoch $e$ at most by time $r$. So, next we can conclude every other honest node must have received $\frac{2}{3}|N|\Signatures$ for epoch $e$ at most by time $r+\Delta$. Hence, no honest node can be in epoch $e'<e$ at most by time $t+\Delta$.
\end{proof}

\begin{lemma}\label{L14} 
Let us assume the following (a) more than $\frac{2}{3} |N|$ nodes are honest, (b) let $[t_{0}, t_{1}]$ denote a period of synchrony. Suppose that at some time $t \in (t_{0} + 1sec+11\Delta, t_{1} ]$, some honest node $n$ is at most $1sec+10\Delta$ into its local epoch $e$. Then, at time $t$ no honest node is in local epoch $e'$ where $e' > e$.
\end{lemma}

\begin{proof}
Now we know node $n$ enters epoch $e$ at time $t-1sec-10\Delta$. Every other honest node must have entered epoch $e$ at most by time $t-1sec-9\Delta$. And no honest node will send $\Timeout$ message by time $t$ because they are not even $1min$ deep in their epoch. Hence at time $t$ no honest node is in a epoch $e'>e$.
\end{proof}
\begin{lemma}[Progress in periods of synchrony]\label{L16} %
(Progress in periods of synchrony). Let us assume the following: (a) more than $\frac{2}{3} |N|$ nodes are honest; (b) time interval $[t_{0}, t_{1}]$ is a period of synchrony; (c) dishonest nodes do not increase the length of voting honest nodes' longest notarized chains via $\Sync$ messages after the timeout block gets proposed and before the voting nodes vote on the proposed timeout block; (d) if a dishonest node chooses to grow a honest nodes' longest notarization chain via a $\Sync$ message before a timeout block is proposed, then it will do that for all the honest nodes before a timeout block is proposed; (e) an honest node $u$ is the proposer of epoch $e$, and it first entered epoch $e$ at some time $r \in (t_{0} + 1sec, t_{1} - 2sec)$; and (f) node $u$ proposes an epoch-$e$ block at some time $t$ where $r+1sec \leq t < t_{1} - 1sec$. Then, by time $t + 3\Delta$,

\begin{enumerate}[label=(\alph*)] 
    \item every honest node will have seen a notarization for a proposed timeout block at most by $t+3\Delta$;
    \item every honest node will have seen a notarization for a proposed normal block at most by $t+2\Delta$;
    \item every honest node must still be in epoch $e$ by $t+11\Delta$;
    \item node $u$ must have proposed more epoch-$e$ blocks by $t+11\Delta$.
\end{enumerate}
\end{lemma}

\begin{proof}
Let $\langle e, s \rangle$ denote the epoch and sequence number of the block $B$ proposed. Claim (d) is directly implied by Claims (a) and (b). Claim (c) follows readily from Lemma \ref{L14}. 

We proceed to prove Claim (a).
For timeout block, consider $s = 1$.
Due to Lemma \ref{F1} and Lemma \ref{L14}, by time $r + \Delta$, all honest nodes must be in epoch $e$; and by time $r + 1 + \Delta $, they must have received the timeout proposal at $\langle e,1 \rangle $.

As soon as node $u$ first enters epoch $e$ (recall that $r$ denotes this time), it waits for $1 sec$ and then proposes block at $\langle e,1 \rangle $. To show that every honest node will receive a notarization for the timeout block at most by $r + 1sec + 3\Delta$, it suffices to show that the parent chain of the timeout block at $\langle e, 1 \rangle $ extends from is at least as long as any honest node’s longest notarized chain when it first enters local epoch $e$. 

We know by time $t$, that is before a timeout block gets proposed every honest node $i$,  has received the parent chain the proposed block extends from and its notarization. Moreover $i$ is in local epoch $e$. It would suffice to show that $i$ will vote for the proposal when this happens (and this holds for every honest node $i$). To show this, it suffices to prove that the parent chain denoted $\chain_p$ the honest proposal extends from is at least as long as any honest node’s longest notarized chain before the honest proposer starts to propose.
 
Now, if the parent chain $\chain_p$ of the timeout block ends at epoch $\langle e', l \rangle $ where $l$ is the sequence number of the last block in the $\chain_p$ for epoch $e' < e$. 
  
When we have a honest proposer and no dishonest node alters the length of longest notarized chain of honest nodes, a honest node moves to the next epoch as soon as it receives at least $\frac{2}{3} |N|$ \emph{signatures} from different voting nodes indicating their wish to move to the next epoch. It can happen in the only scenario when honest nodes' longest notarized chain does not grows for the current epoch in the past $1min$.

So when a honest node's longest notarized chain does not grows in past $1min$, it broadcast a $\Timeout$ message. When a node broadcast a $\Timeout$ message  all other honest nodes at the reception of $\Timeout$ message will send back the set of their longest notarized chains. Hence each honest node must receive the longest notarized chain, since it will broadcast $\Timeout$ message at least once before moving to the next epoch. Hence, before the honest proposer for epoch~$e$ proposes $\langle e,1 \rangle$, all honest node receive the current set of longest notarized chains in existence. Hence, the parent chain, the proposer of epoch $e$ proposes a timeout block from must be at least as long as the longest notarized chain of any other honest node.
  
Hence, all honest nodes will vote will vote for timeout block $\langle e,1 \rangle$. Now, the votes from voting node will reach back to the proposer node at most by $t+2\Delta$. Hence, timeout block for the proposer node will get notarized in $t+2\Delta$ time. Next the proposer node will propose block with epoch $\langle e,2 \rangle$ and will broadcast votes for timeout block $\langle e,1 \rangle$ along with the proposal. It will take at most $t+3\Delta$ for votes to reach the other voting nodes. Hence, the timeout block $\langle e,1 \rangle$ gets notarized for the other honest voting nodes at most by $t+3\Delta$.      
Now, during period of synchrony the timeout block gets proposed and notarized in time at most by $t+3\Delta$ for all honest nodes. And the notarized parent chain $\chain_p$, the normal block gets proposed from must be observed by all honest nodes by time $t+3\Delta$. Now next we need to make sure $\chain_p$ is at least as long as other voting nodes longest notarized chain(if there are multiple, pick one arbitrary). All nodes at time $t+3\Delta$ are still in epoch $e$ because $1min$ has not elapsed since the last notarization. Hence, when honest voting nodes receives proposal for block $\langle e,2 \rangle$ at most by time $t+3\Delta$, the parent chain $\chain_p$, $\langle e,2 \rangle$ extends from must be as long as the any other voting nodes longest notarized chain(if multiple, pick one arbitrarily). Hence all honest voting nodes will vote for $\langle e,2 \rangle$ and will send their votes to proposer node by time $t+4\Delta$. Thus, at time $t+4\Delta$ the normal block $e,2$ gets notarized for the proposer node. Next the proposer node will propose block with epoch $\langle e,3 \rangle$ and will broadcast votes for timeout block $\langle e,2 \rangle$ along with the proposal. It will take at most $t+5\Delta$ for votes to reach the other voting nodes. Hence, the normal block $\langle e,2 \rangle$ gets notarized for the other honest voting nodes at most by $t+5\Delta$.
  
Above $t+3\Delta$ is the time at which the first normal gets proposed and at most by $t+5\Delta$ it will be notarized for all voting nodes. Hence, after the normal gets proposed, it gets notarized for all honest nodes at most by $2\Delta$ time after it gets proposed.
  
Similarly, any normal block proposed by the honest proposer in period of synchrony will get notarized for any honest node at most by $2\Delta$ time after it was proposed by the honest proposer.
\end{proof}

\begin{lemma} \label{L211} Let us assume that more than $\frac{2}{3} |N|$ nodes are honest. Then, at any time, if there exist an honest node with a longest notarized chain of length $l$, then there exist at least one other honest node with a longest notarized chain of length at least $l-1$.
\end{lemma}

\begin{proof}
Suppose a honest node has a longest notarized chain of length $l$ at time $t$. Then, it must have received votes from at least $\frac{2}{3}|N|$ voting nodes out of which at least $\frac{1}{3}|N|+1$ voting nodes are honest. Hence, these honest nodes whose votes led to the notarization of the block at length $l$ must already have a longest notarized chain of length at least $l-1$. 
\end{proof}
 
\begin{lemma} \label{L21} Let us assume the following: (a) more than $\frac{2}{3} |N|$ nodes are honest, (b) let $[t_{0}, t_{1}]$ denote a period of synchrony, and (c) an honest node $i$ enters epoch $e$ at time $t \in [t_{0}+\Delta, t_{1} - 4min -2\Delta]$. Then, during the time interval $[t,t+ 4min + 2\Delta]$, the honest node $i$'s longest notarized chain either grows at least by one new notarized block or the honest node moves to a later epoch.
\end{lemma}

\begin{proof}
Since the honest node $i$ enters epoch $e$ at time $t$, all other honest nodes must have entered epoch $e$ at most by $t+\Delta$. 

For contradiction purpose let us assume during the time interval $[t,t+ 4min + 2\Delta]$, the longest notarized chain of the honest node $i$ neither grows nor the honest node moves to a later epoch.

Let us divide the solution in two sub cases: \textbf{Case (a)}, where node $i$'s longest notarized chain(if multiple pick one arbitrary) is shorter than the longest notarized chain (if multiple pick one arbitrary) of a honest node which has the longest notarized chain among all the honest nodes; and \textbf{Case (b)}, where node $i$ longest notarized chain(if many pick one arbitrary) is as long as any honest node's longest notarized chain (if multiple pick one arbitrary).

\textbf{Case (a)}: 
Since for the sake of contradiction we assume that node~$i$ does not receive any new notarizations and also does not move to a later epoch within the time period $t+4min+2\Delta$, this leads node $i$ to broadcast $\Timeout$ message at time $t+1min$. Once other nodes receive the $\Timeout$ message, they will broadcast $\Sync$ message containing a set of their longest notarized chains. Hence, node $i$ will receive other nodes' longest notarized chains. Thus its longest notarized chain will grow in time $t+ 1min + 2\Delta$. Hence the contradiction proven false.

\textbf{Case (b)}:
By Lemma \ref{L211}, the current case can be further divided into two sub cases: \textbf{Case (b.1)} where node $i$ and at least one other honest node have the same longest notarized chain among all the honest nodes, and \textbf{Case (b.2)} where at least one other honest node has a longest notarized chain that is one block shorter than the longest notarized chain of honest node $i$ and none of the honest nodes have a longest notarized chain that is equal to the length of the longest notarized chain of honest node $i$. 

\textbf{Case (b.1)}:
Again similar to reasons provided in \textbf{case (a)}, node $i$ will broadcast a $\Timeout$ message. Once other honest nodes receive the $\Timeout$ message from node $i$, they will immediately broadcast $\Sync$ message containing the set of their longest notarized chains. So, at most by $t+1min+2\Delta$ all honest nodes must have received the longest notarized chains, longest among all the honest nodes. At most by time $t+1min+2\Delta$ tho honest nodes whose longest notarized chain grew, their $\textit{NotarizationTimer}$ will be reset to zero. Hence, at most by time $t+2min+2\Delta$ either node $i$ adds a new notarized block to one of it's longest notarized chain or else all honest nodes will broadcast a $\Timeout$ message and node $i$ at most by time $t+2min+3\Delta$ will enter a later epoch. Hence, the contradiction is proven false.    

\textbf{Case (b.2)}:
Again similar to reasons provided in \textbf{Case (a)}, node $i$ will broadcast a $\Timeout$ message. Once other honest nodes receives the $\Timeout$ message from node $i$, they will immediately broadcast their set of longest notarized chains. So, at most by $t+1min+2\Delta$ all honest nodes must have received the longest notarized chain that is one length shorter than the longest notarized chain among all the honest nodes. So, at most by $t+1min+2\Delta$ all honest nodes other than node $i$ have a longest notarized that is one block short than the node $i$'s longest notarized chain. 

It can be further divided in two sub cases: \textbf{Case (b.2.a)}, where none of the honest nodes longest notarization grows, and \textbf{Case (b.2.b)}, where other than honest node $i$, at least $(\frac{1}{3}|N|+1)$ honest nodes' longest notarized chain grows by at least one block at most by time $t+2min +2\Delta$. 

\textbf{Case (b.2.a)}: If none of the honest nodes' longest notarized chain grows then at most by time $t+2min+2\Delta$, all honest nodes will broadcast $\Timeout$ message and post it, node $i$ will move to a later epoch at most by time $t+2min+4\Delta$. Hence, the contradiction.

\textbf{Case (b.2.b)}: If at least $(\frac{1}{3}|N|+1)$ honest nodes' (other than node $i$) longest notarized chain grows at most by time $t+2min +2\Delta$. And the honest nodes whose longest notarized chain did not grow at most by $t+2min +2\Delta$ will broadcast a $\Timeout$ messages and will receive the longest notarized chains of other honest nodes at most by time $t+2min+4\Delta$. Hence all honest nodes' will have the longest notarized chain of same length at most by time $t+2min+4\Delta$ as long as the longest notarized chain(if multiple pick one arbitrary) among all the honest nodes.
 
It can be further divided in two sub cases: \textbf{Case (b.2.b.1)}, where none of the honest nodes longest notarized chain grows, and \textbf{Case (b.2.b.2)}, where among all honest nodes (other than the node $i$), at least for $(\frac{1}{3}|N|+1)$ honest nodes', their longest notarized chain grows by at least one block at most by time $t+3min +2\Delta$. 
 
\textbf{Case (b.2.b.1)}: All honest nodes will broadcast $\Timeout$ messages. Node $i$ will move to later epoch at most by time $t+3min+6\Delta$. Hence, contradiction proven false.
 
\textbf{Case (b.2.b.2)}: Node $i$ will broadcast a $\Timeout$ message at time $t+3min$. All other honest nodes will receive it at most by $t+3min +\Delta$.At most by time $t+3min +\Delta$, honest nodes other than node $i$ will broadcast their longest notarized chains. Hence, node $i$'s longest notarized chain will grow at most by $t+3min+2\Delta$. Hence contradiction proven wrong.
 
Suppose other honest nodes' longest notarized chain does not grows by $t+3min+\Delta$ but instead grows post it before time $t+3min+2\Delta$, then node $i$'s longest notarized chain does not grow. Also, all honest nodes other than node $i$ will have the longest notarized chain of same length at most by $t+3min+8\Delta$, which is now one block longer than node $i$'s longest notarized chain. Moreover, all nodes will stay in the current epoch because no node received $\frac{2}{3}|N|$ signatures from different voting nodes for them to move to a later epoch. 
 
Next, node $i$ will broadcast $\Timeout$ message at time $t+4min$. It will receive the longest notarized chain of other other honest nodes at most by time $t+4min+2\Delta$. Hence, the longest notarized chain of node $i$ grows. Thus, the contradiction proven false.
\end{proof}

\begin{lemma} \label{L23} Let us assume the following: (a) more than $\frac{2}{3} |N|$ nodes are honest, (b) timer interval $[t_{0}, t_{1}]$ is a period of synchrony, (c) an honest node $i$ enters epoch $e$ at time $t \in [t_{0}+\Delta, t_{1} - 4 \cdot (4min + 2\Delta)]$. Then during the time interval $[t,t+ 4 \cdot ( 4min + 2\Delta)]$, the honest node $i$'s finalized chain either grows at least by one new finalized block or the honest node moves to a later epoch.
\end{lemma}

\begin{proof}
By Lemma \ref{L21} a honest node either adds at least one new block to it's longest notarized chain in time interval $(4min + 2\Delta)$ or it moves to a later epoch. 

If the honest node $i$ does not move to a later epoch in the time interval time interval $4 \cdot (4min + 2\Delta)$ then it will add at least four new(timeout and three normal blocks) notarized block to it's longest notarized chain.

Hence, $i$ will add a new block to it's finalized chain at most by time interval $t+ 4 \cdot ( 4min + 2\Delta)$.

But, if node $i$ does not add four new notarized blocks within the time interval $[t,t+ 4 \cdot ( 4min + 2\Delta)]$, then by Lemma \ref{L21}, it will move to a later epoch.
\end{proof}
\begin{lemma} \label{L20} Let us assume the following: (a) more than $\frac{2}{3}|N|$ nodes are honest, (b) let $[t_{0},t_{1}]$ denote a period of synchrony, (c) epoch $e$ has an honest proposer, (d) last honest node enters epoch $e$ at time $t \in [t_{0}+\Delta, t_{1}-\Delta]$, then any honest node at time $t + \Delta$ will have a longest notarized chain of same length that is at most one block shorter than any node's longest notarized chain.

\end{lemma}

\begin{proof}
Let us denote a epoch $e'$, which represent the previous epoch from which honest nodes move to a later epoch~$e$.

There are two ways for the longest notarized chain of a honest node to grow (a)it receives a proposal block from the proposer and that gets notarized (b) when it receives a $\Sync$ message. 

Let us consider the following happens in epoch $e'$. For a proposed block to get notarized it must require at least $\frac{2}{3}|N|$ $\Votes$. Hence, the proposed block to get notarization will require $\Votes$ from at least $(\frac{1}{3}+1)|N|$ honest nodes. Let us call this set of $(\frac{1}{3}+1)|N|$ honest nodes as set $A$. Hence, the honest nodes in set $A$ will have longest notarized chains of same length that is at most one block shorter than any node's longest notarized chain. The reason behind one block shorter longest notarized chain for honest nodes as compared to any node's longest notarized chain is because a dishonest proposer might not send the notarization of the last notarized block to some or all of the nodes.  

In epoch $e'$ when a honest node's longest notarized chain does not grows in past $1min$, it will broadcast a $\Timeout$ message. When other honest nodes receive the $\Timeout$ message they will broadcast a $\Sync$ message containing their longest notarized chains. 

Considering node $i$ to be the last honest node to enter epoch $e$ during the period of synchrony at time $t$. Hence, all honest nodes must have entered epoch $e$ during the time interval $[t-\Delta, t]$. So, node $i$ at most by $t$ must have received $\frac{2}{3}|N|$ signatures from different voting nodes, required for it to move to epoch $e$ from the previous epoch. Then, the last $\Timeout$ message which will result in the accumulation of $\frac{2}{3}|N|$ signatures required must have been broadcasted in the time interval $[t-\Delta, t)$. Hence the $\Sync$ message from all honest nodes will get broadcasted at most by $t$. Hence, at most by $t+\Delta$ node $i$ and all other honest nodes will have a longest notarized chain of same length that is at most one block shorter than any node's longest notarized chain.
\end{proof}

\begin{lemma} \label{L22} 
Let us assume the following (a) more than $\frac{2}{3} |N|$ nodes are honest, (b) let $[t_{0}, t_{1}]$ denote a period of synchrony, (c) there are two consecutive epochs $e, e+1$ with honest proposers $P_{1}, P_{2}$ respectively. Suppose that honest proposer $P_{1}$ enters epoch $e$ at time $t \in [t_{0}+\Delta, t_{1}-(2sec+18\Delta+1min)]$, then for every honest node its finalized chain at time $t_1$ must be longer than its finalized chain at time $t$.
\end{lemma}

\begin{proof}[Proof of Lemma \ref{L22}] 
\label{proofL12}
Now the honest proposer $P_{1}$ enters epoch $e$ at time $t$. All other honest nodes must enter epoch $e$ at most by time $t+\Delta$. By Lemma \ref{L20}, at most by time $t+2\Delta$ all honest nodes must have a longest notarized chain of same length that is at most one block shorter than any node's longest notarized chain.

It can be further divided into two cases: \textbf{Case (a)}, where even if the dishonest node has one extra notarization compared to  the honest node's longest notarized chain, the dishonest node will not release it; and \textbf{Case (b)}, where the dishonest node may send notarization for the last block it withheld selectively to some nodes.

\textbf{Case (a)}: In this scenario all honest nodes will have the same set of longest notarizatied chain at most by $t+2\Delta$.  At time $t+1sec$, honest proposer $P_{1}$ proposes a block. All honest nodes must receive the proposed block at most by time $t+1sec+\Delta$. Now the honest nodes will vote for the proposed block if it's longest notarized chain is as long as the parent chain of the proposed block. For any honest node $i$, until the time proposed block gets notarized, their longest notarized chain will not grow any longer than the parent chain of the proposed block because of two reasons(i) no new block gets notarized before the proposed block (ii) No $\Timeout$ message gets broadcasted since it is not more than $1min$ since either the honest node entered a new epoch or since the longest notarized chain grew. Hence, no $\Sync$ message get broadcasted, (iii) we assume dishonest node do not broadcast notarization of the last block if withheld. Hence, by the time any honest node $i$ receives the first proposed block, it's longest notarized chain must be same as the parent chain of the proposed block. Hence, all honest nodes will vote for the proposed block. By Lemma \ref{L16} the timeout block gets notarized for every honest node at most by $3\Delta$ time. Next every normal block proposed for the same reasons mentioned for the timeout block, will get notarized at most by $2\Delta$ (from Lemma \ref{L16}) for every honest node after the normal block gets proposed.

Hence, every honest node will see notarization of a timeout block and 3 normal blocks at most by $t+1sec+9\Delta$. Hence the finalized chain of every honest node will grow by at least one block at most by time $t+1sec+9\Delta$.

\textbf{Case (b)}: In this scenario, we assume dishonest node has one block longer notarized chain compared to  other honest nodes at least till $t+2\Delta$. Post $t+2\Delta$ a dishonest  node can choose to send notarization selectively to some of the nodes.  

To not let honest node's finalized chain grow, the dishonest node can delay releasing notarization of the last block it withheld at least to $(\frac{1}{3}|N|+1)$ honest voting nodes at most by $\Delta$ time before they start to vote. This is because of the reason once the timeout block in a epoch gets notarized for every honest node, post it even if dishonest node releases notarization, the notarized chain timeout block belongs to will be as long as any node's longest notarized chain.

Hence assuming, dishonest node releases notarization of the last notarized block it withheld to at least $(\frac{1}{3}|N|+1)$ honest voting nodes before they receive the proposed timeout block. Now let set $A$ represent set of honest nodes to which the dishonest proposer has released the notarization to. Let set $B$ represent the set of honest nodes to which the dishonest node has not released the notarization to. It can be further divided in two sub cases: \textbf{Case (b.1)}, where the honest proposer for current epoch is from set $B$, and \textbf{Case (b.2)}, where the honest proposer for current epoch is from set $A$.

\textbf{Case (b.1)}: In this scenario, no node from set $A$ will vote for the block proposal. Hence, all honest nodes will send broadcast a $\Timeout$ message at most by $t+2\Delta+1min$. So, all honest nodes will enter epoch $e+1$ at most by $t+3\Delta+1min$. Since nodes in set $A$ have longest notarized chain in existence and the by Lemma \ref{L16} all honest nodes will have the longest chain at most when they are $2\Delta$ deep in epoch $e+1$. Next, since dishonest node has no extra notarization with held, it can not grow the longest notarized chain of any honest node. Hence, like mentioned in \textbf{Case (a)}, at most by $1sec+9\Delta$ after a honest proposer enters a epoch, all honest node add at least one new finalized block to their finalized chain. Hence, at most by time $t+11\Delta+1min+1sec$ every honest node's finalized chain grows.

\textbf{Case (b.2)}: The honest proposer for current epoch is from set $A$. It can be further be divided in two sub cases \textbf{Case (b.2.a)} all the dishonest nodes vote for the proposed blocks so that the proposed blocks gets notarized and which leads ultimately for the finalized chain of honest nodes to grow, \textbf{Case (b.2.b)} some or all of the dishonest nodes do not vote, such that some or all the  proposed block do not get notarized or the dishonest node sends the notarization it wiht held from honest nodes, such that voting nodes receive it after the honest proposer of current epoch proposes a block and before the voting nodes have voted on the proposed block. Hence no honest node will see a new block added to it's longest notarized chain.

\textbf{Case (b.2.a)}: In this case, as show in \textbf{Case (a)} the finalized chain of nodes in set $A$ grows at most by $t+1sec+9\Delta$. The honest nodes in set $B$ will not vote for proposed blocks since their longest notarized was shorter than parent chain of the proposed block. The honest nodes in set $B$ will broadcast $\Timeout$ message at most by $t+2\Delta+1min$. The nodes in set $B$ will receive the longest notarized chain of all other nodes at most by $t+4\Delta+1min$ via $\Sync$ messages. Hence, the finalized of nodes in set $B$ will grow at most by $t+4\Delta+1min$.

\textbf{Case (b.2.b)}: In this case, the dishonest node can chose not to vote any time such that any of the honest nodes' finalized chain does not grows. For this to happen with maximum delay, the dishonest node can avoid voting for the third normal block but still vote for the timeout block and the two normal blocks post the time timeout block. Now from \textbf{Case (a)} we know for the first three block(timeout block, plus two normal blocks) of a epoch to get notarized takes $t+1sec+7\Delta$. Hence, nodes of set $A$ has a three block longer notarized chain compared to  the longest notarized chain of nodes of set $B$. Since the longest notarized chain of nodes in set $B$ does not grows, they will broadcast a $\Timeout$ message at most by $t+2\Delta+1min$. At most by time $t+3\Delta+1min$ all nodes will broadcast their longest notarized chain via $
Sync$ message. Hence, nodes in set $B$ will have their longest notarized chain same as the longest notarized chain for nodes in set $A$. Since dishonest nodes choose not to vote for the third normal block proposed, nodes in set $A$ will not receive a new notarization and will broadcast $\Timeout$ message at most by $t + 1sec + 7\Delta + 1min$. Hence all honest nodes should move to epoch $e+1$ at most by $t + 1sec + 9\Delta + 1min$. 

Now when the honest node enters epoch $e+1$, they all have the same set of longest notarized chains. And the dishonest node is not left with any new notarization's to selectively send to a some nodes. So, similar to \textbf{Case (a)}, the finalized chain of all honest nodes in epoch $e+1$ will grow at most by $t'+1sec+9\Delta$ (where $t'$ is the time when proposer of epoch $e+1$ enters epoch $e+1$), which is equivalent to $t+2sec+18\Delta+1min$.
\end{proof}

\begin{proof}[Proof of \cref{L23L} (Liveness Theorem)]
For every honest node's finalized chain to grow at time $t_1$ as compared to time $t_0$, there must be a proposer in period of synchrony which proposes at least three consecutive normal blocks. And by Lemma \ref{L23} and Lemma \ref{L22}  when this happens the finalized chain grows. And, if the proposer does not propose a block since past $1min$ or if the proposed block does not get notarized since past $1 min$, all honest nodes from $N$ set start broadcasting $\Timeout$ messages indicating their intention to move to the next epoch. A node moves to the next epoch, once the node receives $\frac{2}{3} |N|$ signatures from different voting nodes indicating their wish to move to the next epoch. Before moving to the next epoch without having to finalize a new block, the proposer can propose and get notarized at most three blocks (one time out block, two normal block). In the condition just mentioned, the proposer node can stay for $1sec + 7\Delta + 1min$ in the epoch after first entering the local epoch before moving to the next epoch. Now because of round robin proposer selection policy, also the honest nodes need to be more than $\frac{2}{3}|N|$ and the period of synchrony being of minimum length $(1sec+7\Delta+1min)\cdot(f)+4\cdot(4min+2\Delta)+(2sec+18\Delta+1min)$ , we are guaranteed at least two honest proposer for consecutive epochs in the period of synchrony and by Lemma \ref{L22}, it will lead to the addition of one new finalized block in every honest node's finalized chain at most by time $t_{1}-(2sec+18\Delta+1min)$.
Hence, the finalized chain of all honest nodes at time $t_{1}$ must be longer than time $t_{0}$. 
\end{proof}

\section{Detailed explanation for messages required to finalize a block}
\label{sec:formula}

\subsection{Chained HotStuff~\cite{yin2018hotstuff}}
\subsubsection{Normal mode}
broadcast - replica sends message to itself as well as everyone else 

\begin{enumerate}
    \item{Prepare phase:}
Proposer broadcast $N$ messages and and replicas vote. Requires minimum of $N-f$ votes. Hence total messages $2N-f$.
\end{enumerate}

{Total messages required to finalize a block in non failure mode:} $= 2N-f$

\subsubsection{Failure mode}
In the failure mode we assume leader does not commits anything in the Decide phase. Also upon timeout of a view, the replica sends a message containing generic quorum certificate to the leader of the next view. So, in total it will be $N-1$ messages.

total messages $= f\cdot(2N-f+ N - 1)+2N-f$ = $f\cdot(3N-f-1)+2N-f$

\subsection{ePBFT~\cite{ic_cyber_2019}}
 
\subsubsection{Normal mode}
\begin{enumerate}
\item{Pre-prepare phase:}
Proposer broadcast $N-1$ messages.

\item{Prepare phase:}
Replicas send back message to leader and in total $2f+1$ messages from the replicas are required by leader to make progress.

\item{Commit phase:}
Proposer broadcast $N-1$ messages.
\end{enumerate}

Total messages required to finalize a block $= 2N - 2 + 2f + 1 = 2(N+f) - 1$

\subsubsection{Failure mode}
In the failure mode for maximum messages we assume 
\begin{enumerate}
    \item leader sends Pre-prepare messages to only $2f$ replicas. - $2f$ messages.
    \item $2f$ replica reply to the leader with Prepare-1 message. - $2f$ messages.
    \item leader does not send any commit message.
    \item $2f$ replica broadcast Prepare-2 message. -  $2f(N)$ messages.
    \item timeout happens and replica move to the next view.
    \item let $k$ nodes ask for checkpoint message from $N-1$ neighbor nodes. so check point messages: $k \cdot (N-1)$

\end{enumerate}

Total messages $= f\cdot(2f(N+2))+2(N+f) - 1$ + $k \cdot (N-1)$ .

\subsection{PBFT~\cite{castro1999practical}} \label{pbft_label}
 
\subsubsection{Normal mode}

\begin{enumerate}
     \item{Pre-prepare phase:} Proposer broadcast $N-1$ message.

\item{Prepare phase:}
Replicas broadcast $(2f+1)(N)$ message.

\item{Commit phase:}
Proposer in addition to all voting replicas broadcast $N^2$ messages.
\end{enumerate}

{Total messages required to finalize a block: } $= N - 1 + (2f + 1)N + N^2 = N^2+2N(f+1) - 1$

\subsubsection{Failure mode:}

In the failure mode for maximum messages we assume 
\begin{enumerate}
    \item leader sends Pre-prepare messages to only $2f$ replicas. - $2f$ messages.
    \item $2f$ replicas in prepare mode broadcast message. - $2f(N)$ messages.
    \item leader does not send any commit message.
    \item timeout happens and replica move to the next view
    \item check point messages - a replica periodically at fixed time intervals broadcast  checkpoint message to other replicas. Each replica collects $(2f+1)$ check point messages. Let the number of times a replica  broadcast checkpoint messages be the constant $C$.
\end{enumerate}

Total messages $= f\cdot(2f + 2f(N)) + $C$ \cdot N(2f+1) + N^2+2N(f+1) - 1$.
\subsection{Sync HotStuff~\cite{abraham2020sync}} 

\subsubsection{Normal mode}

\begin{enumerate}
     \item{propose phase:}
Proposer broadcast $N-1$ messages.

\item{Vote phase:}
replica broadcast the proposal to other replica (assuming replica will not send to self and will not send to proposer) - $N-2$. Replica next broadcast vote - $N-1$ vote messages. 

Total $N$ voting replica, hence $(N-1)(N-1)+(N-1)(N-2)$ messages send.
\end{enumerate}

{Total messages required to finalize a block: } $= 2(N-1)^2$

\subsubsection{Failure mode}
In the failure mode, a a view is changed (a)when either a leader proposes less number of blocks in a certain time period, when this happens replica broadcast a blame message or (b) leader proposes equivocating blocks, when this happens leader broadcast the blame messages and as well as the two equivocating blocks 
\begin{enumerate}
    \item blame messages broadcast and equivocating blocks broadcast - $2(N-1)(N-1)$ messages.
    \item quit view - upon receiving $(f+1)$ blame message , broadcast them and quit view - $(N)(N-1)$ messages.
    \item move to next view and send a highest certified block to the new leader - $N-1$ messages.
    \item new leader broadcast new view message, containing highest certified block - $N-1$
    \item replica forward new - view message to all other replica's - $(N-1)(N-2)$
    \item vote - broadcast vote for highest certified block from new - view message - $(N-1)(N-1)$
\end{enumerate}

{Total messages required to finalize a block: }
Total messages $= f \cdot  (2(N-1)^2 + N(N-1) + (N-1) + 2(N-1)^2 )  + 2(N-1)^2 = f \cdot((N-1)(5N-3) )  + 2(N-1)^2$
\subsection{Tendermint ~\cite{buchman2016tendermint}}

Proposer Selection - is weighted( in context to voting ability) round robin. We assume equal weights for all nodes in context to voting ability.
\subsubsection{Normal mode}
\begin{enumerate}
    \item{PROPOSAL message and PREVOTE message} : Proposer broadcast a proposal message to $N-1$ nodes. And a replica upon receiving a PROPOSAL message broadcast a PREVOTE message.

\item{Upon receiving PROPOSAL message and $2f+1$ PREVOTE message} : A replica broadcast PRECOMMIT message. 

\item{Upon receiving PROPOSAL and $2f+1$ PRECOMMIT message} : A replica finalizes.
\end{enumerate}

{Total messages required to finalize a block: }$= N-1+4f+1 = N + 4f$ messages

\subsubsection{Failure mode}
 For worst case scenario, we assume a replica does not finalize upon receiving $2f+1$ PRECOMMIT message. 

Hence, total messages $= f\cdot(N + 4f)+N+4f$ $=f\cdot(N+4f)+N+4f$

\subsection{Pipelet}
 
\subsubsection{Normal mode}

\begin{enumerate}
     \item Proposer broadcast proposal message - $N-1$ which also contains votes for the parent block of the proposed block. When the voting node receive the proposal they notarize the parent block of the current proposed block.

\item{Voting:}
Nodes send vote back to proposer of current epoch. So, $N-1$ messages send back to the proposer of the current epoch.
Hence, the block gets notarized for the proposer.

\end{enumerate}
Total messages required to finalize a block : $2N-2$

\subsubsection{Failure mode:}

In the failure mode for worst case scenario we assume not able to finalize block for $f$ leaders
\begin{enumerate}
    \item Proposer sends proposal messages  $N-1$ messages.
    \item Voting node send back their votes to proposer and next the proposer does nothing. In all proposer receives $N-1$ votes.
    \item Next due to no progress being made, nodes start to broadcast clock messages. Total clock messages broadcasted $N \cdot (N-1)$
\end{enumerate}

Total messages $= f \cdot (N^2+N-2) + 2N-2$.
\\
\section{Notation}
\label{sec:notation}

\begin{table}[h!]
\centering
\caption{Notation}
\label{tab:notation}
\begin{tabular}{ | p{0.42\columnwidth} | p{0.66\columnwidth} |  } 
  \hline
 Symbol & Description  \\ 
  \hline
  \hline
  \multicolumn{2}{|c|}{Constants} \\
  \hline
  $N$ & set of nodes  \\ 
  \hline
   $G^*$ & genesis block  \\ 
  \hline
  \hline
  \multicolumn{2}{|c|}{State of a Node} \\
  \hline
  $e$ & current epoch \\ 
  \hline
  $s$ & current sequence number \\ 
  \hline
  $\textit{EpochTimer}$ & wall clock timer to measure elapsed real time since reset; used by proposer \\
  \hline
  $\textit{NotarizedTimer}$ & wall clock timer to measure elapsed real time since reset; used by voting nodes \\
  \hline
  $\textit{TimeoutSignatures}(e)$ & set of nodes from which $\Timeout \langle e', \Signature \rangle$ messages have been received for epoch $e'$\\
  \hline
  $\textit{Votes}(B)$ & set of nodes from which $\Vote \langle B, \Signature \rangle$ messages have been received for block $B$\\
  \hline
  $\textit{Notarized}$ & set of notarized blocks \\
  \hline
  $\textit{Longest}$ & set of longest notarized blocks \\
  \hline
  $\textit{Finalized}$ & set of finalized blocks \\
  \hline
  \hline
  \multicolumn{2}{|c|}{Messages} \\
  \hline
  $\Timeout\langle e', \Signature \rangle$ & indicates willingness of a node to move to epoch $e'$ \\
  \hline
  $\Vote\langle B, \Signature\rangle$ & indicates a node voting for block $B$ \\
  \hline
  $\Proposal\langle B', \textit{Votes}[B],$ $\Signature\rangle$ & proposes a new block $B'$ and along with it sends aggregate votes for block $B$ and $sign$ for current epoch \\
  \hline 
   $\Sync\langle \textit{Chains}\rangle$ & broadcast the \textit{Longest} set for the current node \\
  \hline
\end{tabular}
\end{table}

\end{appendix}

\end{document}